\documentclass[12pt,reqno]{amsart}

\usepackage{amssymb}
\usepackage{mathabx}
\usepackage{bbold}
\usepackage{mathrsfs}               
\usepackage{dsfont}		
\usepackage{fancyhdr}
\usepackage{a4wide}		
\usepackage[colorlinks=true,linkcolor=blue]{hyperref}%

\usepackage{ulem}
\usepackage[colorinlistoftodos,prependcaption,textsize=tiny]{todonotes}

\usepackage{pgfplots}
\usetikzlibrary{patterns}

\usepackage{mathtools}\usepackage{mathtools}			

\usepackage{bm}		
\usepackage{braket}			

\numberwithin{equation}{section}	 

\usepackage{enumitem}						

\newcommand{\bk}{ { \bm k 	}}
 
\renewcommand{\P}{ {\bf P}  }
\newcommand{\Q}{ {\bf Q}  }
 
\newcommand{\M}{\mathcal M }
\newcommand{\calN}{\mathcal N}

\renewcommand{\d}{\mathrm{d}}							

\newcommand{\N}{\mathbb{N}} 		
\newcommand{\C}{\mathbb{C}}			 
\newcommand{\R}{ \mathbb R  } 			
\newcommand{\RE}{	\mathbb R }

\newcommand{\E}{\mathbb E}
\newcommand{\B}{  \mathbb B}
\newcommand{\I}{\mathcal I}
\newcommand{\V}{\mathbb V}
\newcommand{\F}{\mathscr{F}}
\newcommand{\h}{\t{h}}

\newcommand{\n}{|\hspace{-0.5mm}|\hspace{-0.5mm}|}
\newcommand{\Vu}{  \n V \n_\mu  	}
\renewcommand{\H}{\mathscr H }

\newcommand{\1}{\mathds{1}}

\renewcommand{\le}{\leqslant}
\renewcommand{\leq}{\leqslant}
\renewcommand{\ge}{\geqslant}
\renewcommand{\geq}{\geqslant}     
\newcommand{\vp}{\varphi}
\newcommand{\ve}{\varepsilon}

\renewcommand{\t}[1]{\textnormal{#1}}

\renewcommand{\[}{   [  }

\newcommand{\eff}{\t{eff}}
\newtheorem{condition}{Condition}
\newtheorem{theorem}{Theorem}[section]
\newtheorem{proposition}{Proposition}[section]
\newtheorem{corollary}{Corollary}[section]
\newtheorem{lemma}{Lemma}[section]

\newtheorem{definition}[theorem]{Definition}
\theoremstyle{remark}
\newtheorem{remark}{Remark}[section]

\theoremstyle{definition}

\let\underbrace\LaTeXunderbrace

\begin{document}

	\title[Tracer particle coupled to a scalar field]{Radiative corrections to the dynamics of a tracer particle coupled to a Bose scalar field}
	
	\author{Esteban C\'ardenas} 
	\address[Esteban C\'ardenas]{Department of Mathematics,
		University of Texas at Austin,
		2515 Speedway,
		Austin TX, 78712, USA}
	\email{eacardenas@utexas.edu}
	
	\author{David Mitrouskas}
	\address[David Mitrouskas]
	{Institute of Science and Technology Austria (ISTA), Am Campus 1, 3400 Klosterneuburg, Austria}
	 \email{mitrouskas@ist.ac.at}

	\frenchspacing
	
	\begin{abstract} We consider a tracer particle coupled to a Bose scalar field and study the regime where the field's propagation speed approaches infinity. For initial states devoid of field excitations, we introduce an effective approximation of the time-evolved wave function and prove its validity in Hilbert space norm. In this approximation, the field remains in the vacuum state while the tracer particle propagates with a modified dispersion relation. Physically, the new dispersion relation can be understood as  the effect of radiative corrections due to interactions with virtual bosons. Mathematically, it is defined as the solution of a self-consistent equation, whose form depends on the relevant time scale. 
\end{abstract}
	
	\maketitle

	{\hypersetup{linkcolor=black}
		\tableofcontents}

	\section{Introduction}
		
The study of particles interacting with quantum fields represents a universal challenge across various physics disciplines, ranging from condensed matter systems to fundamental theories like QED. Understanding the behavior of such systems is often approached through the lens of effective properties of independent particles, including effective dispersion relations, mediated interactions and finite lifetimes.

\vspace{2mm}

Such emergent properties can be categorized into two classes based on their association with real or virtual excitations of the field.  Real excitations, exemplified by phenomena like the polaron effect, involve the tangible presence of field excitations surrounding the particle. In the case of the polaron, the particle carries a cloud of real excitations, such as phonons in a crystal lattice, resulting in a reduction of its mobility. Conversely, virtual excitations, as observed in phenomena like the Lamb shift within QED, involve transient fluctuations in the field without the presence of physical excitations. In the Lamb shift, the creation and absorption of virtual photons induces subtle but measurable shifts in the energy levels of an atom. In the physics literature, effects related to virtual excitations are often referred to as radiative corrections.

\vspace{2mm}

In this work, we consider a Nelson type quantum field theory in three dimensions, where a tracer particle (impurity, electron, etc.) is coupled to a Bose scalar field with large propagation speed. The purpose of our analysis is to understand the impact of radiative corrections on the propagation of the particle. To this end, we study the dynamics of initial states with zero field excitations and employ an effective approximation where the field remains in the initial vacuum state and the particle evolves independently but with a modified dispersion relation. The precise form of this effective dispersion relation depends on the relevant time scale of interest. Our main result provides a rigorous norm estimate for the difference between the original dynamics and the effective evolution. We note that our approximation maintains unitarity, thereby precluding any observed decay effects within the considered time scale.

\vspace{2mm}

Let us now turn to the the mathematical model that we consider. 
The Hilbert space of  the system  consists of 
	the tensor product 
	\begin{equation}
		\textstyle 
		\mathscr{H} : =  L_X^2 \otimes \mathscr{F}
		\qquad 
		\t{with}
		\qquad 
		\mathscr{F} 
		: = 
		\C \oplus \bigoplus\limits_{ n=1}^{\infty} \bigotimes\limits_{\rm sym}^n  L^2(\RE^3) \ . 
	\end{equation}
	Here, $L_X^2 \equiv L^2(\RE^3,dX)$ is the 
	state space for the tracer particle, 
	and $ \mathscr F$ denotes the bosonic Fock space. As usual, 	 we equip $\mathscr{F}$
	   with  standard 
	creation and annihilation operators 
	that satisfy the Canonical Commutation Relations (CCR) in momentum space 
	\begin{equation}
		[	 b_k  , b_\ell^*	] = \delta( k - \ell  ) 
		\qquad
		\t{and}
		\qquad 
		[b_k , b_\ell] = [b_k^* , b_\ell^*]=0  
	\end{equation}
for any $ k , \ell \in \RE^3$, and we denote by $\Omega = (1, {\bf 0}) \in \mathscr{F}$
	the vacuum vector. 
	On the Hilbert space $\H$
	we consider  the dynamics that is generated 
	 by the following Hamiltonian  
	\begin{equation} \label{eq:Ham}
H 
  \  : =  \ 
 \,   - \Delta_X
 \, +  \, 
\mu   \, 
\int_{\RE^{3}} \omega (k) b_k^* b_k \d k 
 \ +  \  \mu^{1/2}  \,  
\int_{\RE^{3}}
	V(k) 
\big( 
e^{i k \cdot X}  
 \, b_k + 
e^{ - i k \cdot X}  
\, b_{k}^* 
\big)   \d k  
	\end{equation}
and its dependence with the large parameter 
\begin{align}
\mu \gg 1.
\end{align}
 Here, the first two terms
correspond to the kinetic energy of the
tracer particle and the field energy, respectively, 
and  the  third term corresponds to their interaction. 
We refer to  $ \omega(k)$ as  the  dispersion relation of the field, 
	and $V(k)$  as the  interaction potential, also called the form factor.
	Both are       real-valued and  rotationally symmetric. 

\vspace{2mm}
 	
	Let us note that the described model is translation invariant, which will be important for our analysis. 
	Throughout this article, we use the following notation
	to denote the operators of momentum
	\begin{equation}
		p : = - i \nabla_X
		\qquad 
		\t{and}
		\qquad 
		P : = 
		\int_{\RE^3 } k \, b_k^* b_k \d k \   \ . 
	\end{equation}
	In particular, translation invariance implies that the total momentum is a conserved quantity, i.e.
	$[H , p +  P ]=0$.  
 Unless confusion  arises, we also use the letter 
		$ p \in \RE^3$ to denote the 
		  Fourier variable  in momentum representation
		  	  $\hat \vp(p) \equiv  (2\pi)^{-3/2  }  \int_{\RE^3}  e^{-i p \cdot X   } \vp(X)  \d X$.
	Finally,  we   denote by 
	$
	\calN  : = \int_{\RE^3 } 
	b_k^* b_k \d k   
	$
	the number operator associated to the boson field.

	\vspace{3mm}
	
We study the  dynamics of the wave function $\Psi (t) = e^{-itH}\Psi_0$ 
for  initial states
	of the form $\Psi_0 = \vp \otimes \Omega$, where $\vp	\in L_X^2$ is suitably localized in momentum space. 
	For concreteness, we assume   that the dispersion is given by 
	\begin{equation}\label{def:dispersion}
		\omega(k) =    \sqrt{  k^2 + m^2}    \ , 
	\end{equation}
where  $m\ge 0$ represents a possible mass term.
	We say that the field is massless if $m=0$, and  say it is massive otherwise.
	As for the interaction, the main example  of physical interest that we keep in mind  corresponds to the Nelson model with ultraviolet cut-off. More precisely,
	\begin{equation}\label{eq:Nelson:form:factor}
		V(k)  =  
		\frac{ \1 (|k|\leq \Lambda )}{	 \sqrt{\omega(k)}	}
	\end{equation}
	for some  positive parameter $\Lambda$.

\vspace{2mm}

To our knowledge, the described scaling for $H$ with $\mu\gg1$ was introduced by Davies \cite{Davies}, who interprets $\mu^{-1}H$ as a model for a heavy tracer particle weakly coupled to the Bose field. Hence, this choice of scaling is sometimes referred to as a \textit{weak-coupling limit}. In \cite{Hiroshima,Hiroshima2}, Hiroshima studies the same scaling for the Nelson model, but with the cutoff $\Lambda(\mu) \to \infty$ removed as $\mu\to \infty$. For a comparison of these works with our results, see Section \ref{subsection:comparison}. In \cite[Section 6.3]{Stefan}, Teufel explains that $H$ can be viewed also as the canonical quantization of a classical system for a particle coupled to a scalar field with propagation speed tending towards infinity. Thus, we think of the boson field as the fast subsystem relative to the tracer particle. We are interested in understanding the effective dynamics that such system may generate.

	\subsection{Description of main results}
	Let us informally describe our main results. 
	These are classified
	according to the   absence
	or existence
	of  a mass term. 
 In order to keep  the following exposition simple,    
	we assume here that the form factor is given by the  Nelson model \eqref{eq:Nelson:form:factor}. 
	
\vspace{2.5mm}

	\textit{Massless fields}.  
In the present case, 
we prove in Theorem \ref{thm 1}  
	  that 
	there is an effective $\mu$-dependent
	Hamiltonian $\h_{g}(p)$ on $L_X^2$
	such that 
	\begin{equation}
		\Psi(t)    \, \approx  \, 
		e^{ - it \h_{g}(p)} \vp \otimes \Omega   
		  \qquad
		\t{as}
		\qquad 
		\mu \rightarrow \infty 
	\end{equation}
for all times $|t| \ll \mu^{1/2} ( \log \mu )^{-1}$. The approximation holds in the norm of $\mathscr{H}$ with explicit error control. 
The effective Hamiltonian  can be written as  a perturbation of the free kinetic energy $ \h_g(p)  = p^2 - g(p)$, 
	and is defined in terms of a function $g(p)$ that solves the following nonlinear equation
	\begin{equation}
		g  (p) 
		=    
		\int_{\RE^3} \frac{\mu \,  |V( k ) |^2\, \d k  	}{ 
			\mu\, \omega (k) + (p-k)^2 - p^2 + g(p)  	 } \  .
	\end{equation}
As we will see,   $g(p)$ 
is an  $O(1)$ correction, caused by the emission and absorption of virtual bosons.

	\vspace{2.5mm}
	
	\textit{Massive fields}.  
	In this situation, the   previous  approximation can 
	be extended to longer time scales by 
	introducing  higher-order radiative corrections. For concreteness, assume $m\ge 1$ for the moment.
We prove in Theorem \ref{thm 2}  that   for all $ n \in \N$ odd there is an effective Hamiltonian $\h_{g_n}(p)$
	such that  
	\begin{equation}
		\Psi(t)   
		  \, \approx   \, 
		e^{ - it \h_{g_n }	(p)	} \vp \otimes \Omega   
	  \qquad
\t{as}
\qquad 
\mu \rightarrow \infty 
	\end{equation}
 for all times   $|t|  \ll \mu^{(n+1)/2}  (\log (\mu))^{ - (n+2)/2}  $.
	The approximation also holds in the norm of $\mathscr{H}$ 
	with an explicit error.  
	The Hamiltonian is defined analogously $\h_{g_n} (p) = p^2 - g_n(p)$, 
	but the  nonlinear expression that defines  $g_n(p)$
	is more involved since 
	 it  contains higher-order   corrections. 
	See Definition \ref{definition gn}. 
		Let us stress that the addition of higher-order corrections
	does not reduce  the distance between
	the original and the effective dynamics, 
	which is always at least of order $O(\mu^{-1/2})$.
	Rather,  	it has the effect of extending    the 
	\textit{time scale}
	for which the approximation is valid.

	\subsection{Discussion}
 Let us briefly comment on the novelties of the present paper.

First, 
		to the best of our knowledge, 
	the derivation of  the   generator
	$g(p)$ 
	as an effective dynamics of  Nelson type models  \eqref{eq:Ham} is new.
The latter is determined through a \textit{nonlinear equation}, 
 and  is reminiscent of   Dyson-type equations
in  quantum field theory that  appear due to one-loop 
	propagator renormalization (resummation of self-energy diagrams); see e.g \cite{FetterWalecka}. 
From this point of view, our approach
	 is  different than previous works  
	that  have analyzed similar problems in a linear fashion.
	
	Second, 	our proof is based on a systematic expansion of the dynamics and suitable error estimates. The expansion utilizes an 	 \textit{integration-by-parts formula}, which
	  exploits  the  phase cancellations due
	 to  fast oscillatory phases in the limit $\mu \rightarrow \infty$. The effective generator $g(p)$ arises from those terms in the expansion that lack such phase cancellations. Combined with the nonlinear aspect of our approach, 
we are able to  gain  control over longer time scales  (and in a fully quantitative setting) in comparison with previous    results.

We  extend the present  discussion 
and compare our results with
related works in Subsection \ref{subsection:comparison}. 

\vspace{2mm}

\textbf{Organization of the article}.\
In Section \ref{section main results} we 
state our main results, Theorems \ref{thm 1} and \ref{thm 2}. 
 In Section \ref{section prelim}
we introduce 
a
suitable 
   integration-by-parts formula 
that allows us to compare the microscopic dynamics with the effective evolution. 
In Section \ref{section massless} we prove Theorem \ref{thm 1} for massless fields, and in Section \ref{section g} we analyze the effective generator in the massless case.
In Section \ref{section massive} we prove Theorem \ref{thm 2} for massive fields, and in
 Section \ref{section resolvent estim} we establish some necessary resolvent estimates.

	\section{Main results}
	\label{section main results}
	In this section, we give the statements of our main results. 
	These correspond to Theorem \ref{thm 1} and  \ref{thm 2}. 
	First, let us give the  precise conditions  that we consider
	on the initial datum $\vp$ and the 
	form factor  $V$.
	
	\begin{condition}
		\label{condition 1}
		The initial datum  is a tensor product 
		$\Psi_0 = \vp\otimes \Omega$ with $\|\varphi\|_{L^2} = 1 $ and $\Omega $ the Fock space vacuum.
		We assume  that there
		exists $ 0 < \ve <  \frac{1}{2}$
		such that  
		\begin{equation}
			\vp  =  \1 (| p | \leq \ve \mu		) \varphi \ . 
		\end{equation}
	\end{condition}
	\begin{condition}
		\label{condition 2}
		The form  factor  $  V: \RE^3 \rightarrow \RE$
		is rotationally symmetric,  non-zero, 
		and we assume that 
		both
		$\|	   V\|_{L^2}$
		and 		$\| |k|^{- \frac{1}{2}}   V		\|_{L^2}$
		are finite. 
	\end{condition}

	\textit{Convention}.
	We say $C >0$ is a \textit{constant} if it
	is a positive number, depending only on $\ve $, $\|	V\|_{L^2}$
	and   $ \| |k|^{-\frac{1}{2}}	V\|_{L^2}$.
	Throughout proofs, its value may change from line to line.

	\vspace{3mm}
	
	Under these conditions, 
	the Hamiltonian $H$ is self-adjoint in its natural domain.
	Indeed, let us  denote here and in the rest of the article
	\begin{equation}
		T  = \mu  \int_{\RE^{3}}	\omega (k) b_k^*b_k \d k 
		\qquad 
		\t{ and }
		\qquad 
		\V 
		\label{V term}
		 	=  
		\mu^{1/2} \int_{\RE^3}    
		V(k) 
\big( 
		e^{i k \cdot X}  
		b_k + 
		e^{ - i k \cdot X}  
		b_{k}^* 
\big) 
 \d k    \ . 
	\end{equation}
Then, it is 
	straightforward to verify that $\V$ is 
	relatively bounded with respect to $T$, with relative bound less than one.
	Thus, the Kato-Rellich theorem implies that   $H$ is self-adjoint 
	on  $D(p^2)\cap D(T)$. 
	Consequently, the  original dynamics
	\begin{equation}
		\Psi(t) = 
		e^{ - i tH } \Psi_0
		\qquad 
		\t{with}
		\qquad 
		\Psi_0 =	\vp \otimes \Omega   
	\end{equation}
	is well-defined for all $t \in \RE$. 
	
	\vspace{2mm}
	
	On the other hand,  given a  bounded  measurable  function $  g:\RE^3 \rightarrow\RE $ 
	let us   define the effective Hamiltonian as 
	\begin{equation}
		\h_g (p)  : = p^2 - g(p)  
	\end{equation}
	with $p=-i\nabla_X$. 	  We refer to $\h_g(p)$ as the \textit{effective dispersion relation} or the \textit{effective generator}.
	Thanks to the boundedness of $g$, 
	the effective generator is again self-adjoint on $D(p^2)$. Thus,  the effective dynamics 
	\begin{equation}
		\Psi_{\eff}(t) : = 
		e^{	- i t \h_g (p) }\vp \otimes \Omega 
	\end{equation}
	is well-defined for all $ t\in \RE$. 
	As we have already explained in the introductory section, the function $g(p)$
	is determined by a nonlinear self-consistent equation. It depends on the choice of $\omega$ and $V$ and, in particular, also
	on the time scale that one is interested in. 
	The available time scales depend on the absence or presence of 
	a mass term  in the dispersion relation of the field.

	\begin{remark}
Let us explain the heuristics behind Condition \ref{condition 1}: For  an excitation of momentum $ k \in \RE^3 $  to be   
emitted  from  an electron  with momentum $ p\in \RE^3$, 
conservation of kinetic  energy reads
\begin{equation}
\label{energy:conservation}
 \mu 	\omega(k) =  p^2 - (p-k)^2  \   .
\end{equation} 
The collection $(p,k ) \in \RE^6$ that satisfy  \eqref{energy:conservation}
are sometimes  referred to as   resonant sets. 
In particular, under Condition \ref{condition 1}
  these sets are excluded.
Thus,  we can interpret the assumption $|p| \leq \ve \mu $
as an energetic constraint. Namely, that the electron lacks sufficient energy to create real excitations. On the level of the effective dynamics, this is manifest in the fact that $\calN \Psi_{\eff}(t)=0$. Our analysis will show that real excitations are in fact suppressed in the original time evolution for large $\mu$, in the sense that $\Psi(t) \approx \Psi_\eff(t)$. 
On the other hand, 
thanks to the Heisenberg uncertainty principle, energy conservation can be
violated by the fluctuations of the field.
These are called
   {virtual} excitations,
   and they  give rise  to non-trivial   modifications of  the dispersion of the electron.
   These effects  are 
    described by the function $g(p)$. 
	\end{remark}

	\subsection{Massless fields}
	Throughout this subsection, we assume that the boson field  is massless. Namely,  that 
	\begin{equation}
		\label{not gap}
		\omega(k) =  
		 |k|
	\end{equation}
	for all $ k \in \RE^3$.  
	The main result in the present  case  is the estimate contained  in  Theorem \ref{thm 1}, stated below.
	
	\vspace{2mm}
The precise definition of the effective Hamiltonian is as follows.

	\begin{definition}
		\label{def g1} We define  
		the function 
		$g : \RE^3 
		\rightarrow ( 0 , \infty)  $ 
		as the solution of the equation
		\begin{equation}
			\label{def g}
			g  (p) 
			= 
		 	\int_{\RE^3} \frac{  \mu  \,  | V(k)|^2  \d k 	}{ 	
				\mu \,  \omega(k) + (p-k)^2 - p^2 + g(p)  
					 }  
		\end{equation}
		for $| p |	< \frac{1}{2}  \mu  $  and $ g \equiv 0$ otherwise.
	\end{definition}

	Let us    argue  that 
		the solution to \eqref{def g} 
	exists,  is unique and of class $C^1$ on $ \{   |p| < \frac{1}{2}\mu			\}$. 
	
\begin{proof}[Sketch of proof]
	Fix $0< \ve_0<1/2$. Then, for   $| p| \leq \ve_0 \mu $   we  define the following auxiliary function  
	 \begin{equation}
	 	\textstyle 
	 	\label{F1}
	 	F( p, x) : =   
	 	\mu 
	 	\int_{\RE^3 }  
	 	| V (k) |^2 
	 	\big(
	 	\mu \, 	\omega(k)   - 2 p\cdot k + k^2  + x 
	 	\big)^{-1 } \d k  \ , 
	 	\qquad x>0 \ . 
	 \end{equation}
	 First, observe    that
	 $	 \mu \omega(k)   -  2 p\cdot k + k^2  + x   \geq (  1 - 2 \ve_0) \mu |k |  $.
	 Hence, $F (p,x) \leq    C $
	 for a constant $C>0$
	 and  so $F$ is well-defined.
	 Next,
	 note        $   F (p,0)  >0  $, 		$x\mapsto F(p,x)$  is $C^1$, 
	 and $\lim_{x \rightarrow \infty} F(p,x )= 0$. 
	 Thus, the intermediate value theorem implies 
	 there exists  $x \in (0, \infty )$  such that   $
	 x - F(p,x) = 0 $.
	 Further, note    
	 $  \partial_x F <0$. 
	 Thus,   $ x\mapsto F (p ,x )  $ is strictly decreasing and the fixed point is unique.
	 We then set $g(p) =   x . $ 
	 Finally,  since  $x \mapsto F(p,x) $ is $C^1$ on  $ |p|< \ve_0\mu$ 
	 the implicit function theorem implies that $p\mapsto g(p)$ is $C^1$ on $\{ |p |  < \ve_0 \mu		\}$.
	 The claim now  follows from taking $\ve_0$ arbitrarily close to $1/2$. 
\end{proof}

	\vspace{1mm}

To state our main results, we need to introduce the following 
 $\mu$-dependent norm
	\begin{equation} 
\label{nu norm} 
		\n	   V		 \n_{\mu} 
		: =    
		 \bigg\|       \frac{V}{\omega_\mu }  \bigg\|_{L^2}  
		\qquad
		\t{with}
		\qquad
		\omega_\mu  (k)  : =  \omega(k) + \mu^{-1 }   \ . 
	\end{equation}
As we shall see below, the  $\mu$-norm \eqref{nu norm} appears naturally 
thanks to   the presence of the generator $g(p)$ in certain denominators. It can be understood  as an infrared regularization  
of the norm
$\|	 \omega^{-1} V	\|_{L^2}$. In particular,
 while for the massless Nelson model \eqref{eq:Nelson:form:factor} the latter norm is infinite, the $\mu$-norm is finite and grows logarithmically with $\mu$:
	\begin{align}
	\label{mu norm nelson}
	\n		V \n_\mu^2 =    \int_{	 | k | \leq \Lambda	 } 		 \frac{\d k }{ |k| (  	 |  k| +  \mu^{-1} )^2 	}  
	= 
	4 \pi \bigg(
	\log \big( \mu \Lambda +1 \big) -  \frac{ \mu \Lambda  }{  \mu  \Lambda + 1 }   \bigg)  .
\end{align}

\vspace{2mm}
		
	We are now ready to state our first main result.

	\begin{theorem} 
		\label{thm 1}
		Let $\omega(k) = |k|$ 
		and assume  
		that $\Psi_0$ and $V $
		satisfy Condition \ref{condition 1} and \ref{condition 2}, respectively. 
		 Then, there exists a constant $C>0$ such that 
		\begin{align}\label{thm:1:bound}
			\|	
			\Psi(t)	-
			\Psi_{\eff} (t)  
			\|_{\mathscr{H}}
			& 	\leq 
			C   \bigg(	 
			\frac{	\n    	  V	\n_\mu    }{\mu^{1/2}		}	 + 
			|t|  \, 
			\frac{	\n    	V	\n_\mu^2   }{\mu^{1/2} }	 
			\bigg) 
		\end{align}
		for all $t \in \RE$ and all   $\mu > 0 $ large enough. 
	\end{theorem}
		
\begin{remark} A few remarks are in order.

\begin{enumerate}[leftmargin=*]
		\item The relevant parameters in \eqref{thm:1:bound} are $\mu$ and $t$. Since $\Psi(t)$ and $\Psi_\eff(t)$ are both normalized, we are interested in the parameter regime, for which the upper bound is small compared to one. 
		
		\vspace{1mm}
	
\item For the  	the massless Nelson model, 
	 the logarithmic growth   \eqref{mu norm nelson} and 
		Theorem \ref{thm 1} imply that we have  an approximation 
		 for $\mu \gg 1$ and 
$$	 |t| \ll \frac{\mu^{1/2}}{\log\mu}	 \ . $$
		We show in Appendix \ref{section mild} how to extend
		this time scale by a factor $(\log \mu )^{\frac{1}{2}}$. 

\vspace{1mm}

\item  It is   possible
		to introduce a scale of less singular of form factors:  $V_a(k) = |k|^{-a } \1 (|k |\leq \Lambda)$
		where $a \in [0,1/2)$.
		For these models, the norm 
		$\n	V_a\n_\mu $
		is uniformly bounded in $ \mu>0$, hence Theorem \ref{thm 1} provides an approximation for $|t|\ll \mu^{1/2} $.
		 In Appendix \ref{section mild} we show how 
		to extend this approximation to times $|t| \ll \mu^{ 1 - a}$. 
\end{enumerate}
\end{remark}

\subsubsection{Polynomial generators}
	In Section  \ref{section g} 
we analyze 
the  generator $g(p)$ in more detail.  
As a first step, we show that for all $ 0 \leq \ve < 1/2$
there are    constants $C_1$ and $C_2$   such that
\begin{equation} 
	C_1 
	\leq 
	g(p)
	\leq 
	C_2  \qquad \forall |p| \leq \ve \mu \ . 
\end{equation}		
We then refine the analysis, and    show   that 
$g(p)$ 
can be  approximately solved   in terms of $|p|$ 
and the value $g_0 := g(0)$.
Namely,  we consider the function  
\begin{equation}\label{eq:approx:sol:g} 
	g_{\eff}(p)
	:= 
	\frac{ \mu }{2 |p|}
	\int_{\RE^{3}} \frac{|V(k)|^2}{|k|}
	\tanh^{-1} \bigg(	 \frac{2 |p| |k |}{	 \mu |k|+ k^2  + g_0		}		\bigg) \  \d k      
\end{equation}		 
and  prove in Proposition \ref{prop expansion g} 
that for every $ 0 \le  \ve  < \frac{1}{2}$
there is a constant $C>0$ such that 
	\begin{equation}
|
		g(p)	
		-
		g_\eff(p) 
		| 
\, \leq \, 
		 \frac{C  \Vu^2 	}{\mu}    
		 \qquad 		\forall |p| \leq \ve \mu \ . 
 \end{equation}
Thus, we can replace $g(p)$ with  $g_\eff (p)$   in Theorem \ref{thm 1}  and  keep  the same error estimate.
	See Corollary \ref{cor:explicit:g} below.		 

\vspace{2mm}

While  
$g_\eff(p) $  is explicit,   it can still be regarded a complicated function  of $p$. 
 To obtain a   simpler  form for the effective dispersion, we  use the series expansion of the hyperbolic tangent function. This motivates the definition of the following polynomial generators
 of degree $N \in 2 \mathbb N _0$ 
\begin{align}\label{eq:hgn}
	\h_g^{(N )}  : = p^2 - \sum_{ j \in  2\N_0 : j\le N} \alpha_j(\mu) |p|^j  
\end{align}
where $\alpha_j (\mu )$
are positive coefficients, 
defined explicitly  in \eqref{series g}. 
For completeness, we also set  $\h_g^{(\infty)} : = p^2 - g_{\eff}(p)$. For example, the generator $\h_g^{(2)}= (1-\alpha_2) p^2 - \alpha_0$, describes a free particle with enhanced mass. For $ N \ge 4$, the generators encompass non-trivial corrections to the parabolic shape of the free dispersion relation. 

\vspace{2mm}

The next corollary shows that for suitable initial conditions, the effective generators $\h_g^{( N )}(p)$ can still be used to approximate the wave function $\Psi(t) = e^{-itH}\varphi \otimes \Omega$.

\begin{corollary}\label{cor:explicit:g} 
	Assume Conditions \ref{condition 1} and \ref{condition 2} and additionally  
	$\varphi \in \mathds 1(|p|\le P_0) L_X^2 $.
	Then 	there is a constant $C>0  $ such that 	
		\begin{align}\label{eq:explicit:gen:approximation}
		& 	\|\Psi(t) - e^{-i t \h_g^{( N )} }\varphi \otimes \Omega \|_{\mathscr{H}} 
		\le 
		C 
		\bigg( 
		\frac{\n V \n_\mu}{\mu^{1/2}} + |t|
		\frac{ \Vu^2 }{\mu^{1/2}}
		+ |t| \bigg( \frac{P_0}{\mu} \bigg)^{N+2}\,	
		\bigg) 	 & & \t{ for }  N   \in 2 \N_0 	\ , 	\\
		& 	\|\Psi(t) - e^{-i t \h_g^{( \infty )} }\varphi \otimes \Omega \|_{\mathscr{H}} 
		\leq   C 
		\bigg( 
	\frac{\n V \n_\mu}{\mu^{1/2}} + |t|
	\frac{ \Vu^2 }{\mu^{1/2} }
	\bigg) 
		& & \t{ for }    N  = \infty \label{eq:explicit:gen:approximation:infinite} \ 
	\end{align}
for all $t\in \RE$, $P_0 >  0 $ and all $\mu >0 $ large enough.
\end{corollary}

\begin{remark}\label{rem:non-conv}
Using the truncated version of the effective dynamics leads to the additional error term   $|t| (P_0/\mu )^{ N +2}$. 
We now   discuss its consequences 
	 for the Nelson model  \eqref{eq:Nelson:form:factor}.

	\begin{enumerate}[leftmargin=*]
		\item Clearly, 
		the new error term imposes additional 
		constraints 
		on the validity of the approximation, 
		when compared to  Theorem \ref{thm 1} or
		 \eqref{eq:explicit:gen:approximation:infinite}.	Indeed, while the original error is uniform in the momentum scale   $P_0$, 
	the new error is not. 
			Observe for instance that at scales    $P_0\sim \mu$   
		the right hand side of  \eqref{eq:explicit:gen:approximation} is small only  for $|t|\ll 1$, as opposed to $|t| \ll \mu^{1/2} (\log \mu)^{-1}$.
		
		\vspace{1mm}
		
		\item 
			In Section \ref{section:range} 
			we compare in detail  the quality of the different approximations, 
		and 
		argue that the new error in \eqref{eq:explicit:gen:approximation} is  in fact \textit{optimal}. 
	For illustration, 
	let  $N =2$ and 
	consider initial data satisfying 
	\begin{align} 
\textstyle	\varphi \in \mathds 1(\tfrac12 P_0 \le |p| \le P_0)L_X^2 \quad \t{with} \quad  		P_0 = \mu^{23/24} \ . 
	\end{align}
	Then,  the additional error is small only for $|t| \ll \mu^{1/6}$. 
	In	Theorem \ref{thm:non-convergence} we strengthen this observation by showing that one  cannot consider longer time scales
		and still expect the approximation to be valid. 
		Namely,  we prove  that there is a constant $\tau>0$ such that convergence fails for $t = \tau \mu^{1/6}$:
		\begin{align}
			\liminf_{\mu \to \infty} \|\Psi(t) - e^{-i t \h_g^{(2)}(p) }\varphi \otimes \Omega \|_{\mathscr{H}} > 0 \ . 
		\end{align}
		On the other hand, Corollary \ref{cor:explicit:g} implies that the approximation remains effective on this time scale for $N \ge 4$. This highlights the relevance of the higher-order corrections in $\h_g^{( N )}(p)$. 
	\end{enumerate}

\end{remark}

	\subsection{Massive fields}
	Let us now state our main result regarding massive boson fields with dispersion relation
	\begin{equation}
		\label{gap}
		\omega(k) =  \sqrt{k^2 + m^2 }
	\end{equation}
	for some mass $m>0$.
	In contrast to the massless case, here one is able to systematically extend the time scale of validity of the effective approximation.
	However, one needs to modify the generator $g(p)$
	to include higher-order terms.

	For the precise definition, we need the following notation
	for a collection of sequences $\sigma$ of length $  j   \in \mathbb N $: 
	\begin{align}
		\Sigma_0  ( j  ) 
		& 	: = 
		\textstyle 
		\{	   \sigma \in \{ 	 +1  , -1 \}^j  \ :  \  	\sum_{ i  = 1}^{\ell }	\sigma (i) \geq 1 \ \forall \ell \leq j    -1  \ \  \t{and} 	  \ \   \sum_{ i=1 }^j \sigma (i) = 0  			\}   \ . 
	\end{align}
See Figure 1 for a  visual  representation  of an element $\sigma \in \Sigma_0(8)$. 
	We will also
	denote
	$b_k^{\#_{ + 1 } } \equiv b_k^*$ and $b_k^{\#_{-1 }	} \equiv  b_k$.

	\begin{definition}
		\label{definition gn}
		Let $n \in \mathbb N $.
		Then, we define  the function 
		$g_n : \RE^3 \rightarrow ( 0 , \infty )  $
		as the solution of the equation 
		\begin{align} 
			\nonumber 
			\label{def gn}
			g_n (p)
			= 
			\sum_{	\substack{	 j =2 	 \\  j \t{ even }	 }	}^{n+1 } 
			\sum_{\sigma \in \Sigma_0(j)} & 
			\mu^{j/2}
			\int_{\RE^{3  j }} 
			\d k_1 \cdots \d k_{j}
			V(k_1) \cdots V(k_{j} )
			\, 
			\langle \Omega , b_{k_{j}}^{\#_{\sigma(j )}} \ldots  b_{k_1}^{\#_{\sigma(1)}} \Omega\rangle  \\
			\textstyle 
			& 	     \times   
			\prod_{ \ell =1}^{ j  -1  } 
			\bigg( 
			 \mu 
			\sum_{ i =1}^\ell \sigma(i) \omega(k_i) 
			+ 
						\Big( 	 p - \sum_{i=1}^\ell \sigma(i) k_i		\Big)^2
						- 
			p^2 + g_n(p)  
			\bigg)^{-1} \ 
		\end{align}
		for $|p|  <  \tfrac12 \mu$ and
			$g_n \equiv  0 $ otherwise. 
	\end{definition}

	\begin{remark}
		The generator $g_n(p)$  
		can be regarded as the  solution of  the fixed point equation 
		\begin{equation}
\label{gn:fixed}
			g_n(p)
			 = 
			   F_n( p , g_n(p)) \ , \qquad |p| <  \frac{1}{2} \mu 
		\end{equation}
		where   $F_n(p,x)$ is
		determined from  the right hand side of \eqref{def gn} 
		(see e.g. \eqref{Fn}). 
	In particular,  for each $\sigma \in \Sigma_0(j)$, 
	 the bosonic expectation values
		can be evaluated   in 
		terms of $\delta$-distributions via Wick's theorem. 
		After carrying out  the contractions, 
	the    summability condition of  $\Sigma_0(j)$  
	implies that every  denominator in  \eqref{def gn}  
		is   manifestly positive and decreases  with $x$.
		Thus,  the proof  we sketched  
	for  the $n=1$ case below Definition \ref{def g1}
		can be adapted to   $n\geq 2$
		to show that the solution to \eqref{gn:fixed} exists,  is unique and of class $C^1$. 
		We leave the details to the reader. 
	\end{remark}

	\begin{figure}[t]
		\label{fig2}
		\includegraphics[width=10cm]{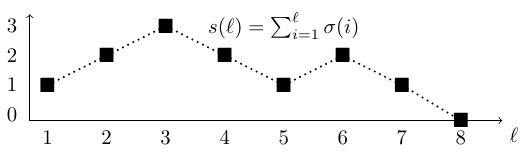}
		\centering
		\caption{The  partial sums of 
			$\sigma = (1,1,1,-1, - 1, 1,-1, - 1)\in \Sigma_0(8)$}
	\end{figure}

	We  now consider the effective dynamics
	\begin{equation}
		\Psi_{\eff}^{(n)} (t) : = e^{	 - i t \h_{g_n}(p)	}\vp \otimes \Omega 
	\end{equation}
	for all $t \in \RE$ with $\h_{g_n}(p) = p^2 -g_n(p)$.
	
	\vspace{2mm}
	
	In the next result, we keep track of  the role of the mass term $m$. 
	We recall \eqref{nu norm} for the definition of $\Vu$.

	\begin{theorem} 
		\label{thm 2}
		Let $\omega(k) = \sqrt{k^2+m^2}$  
		and assume  
		that $\Psi_0$ and $V $
		satisfy Condition \ref{condition 1} and \ref{condition 2}, respectively.  
		Then,  there  exists a   constant $ C > 0  $ 
	such that 	for every $n\in \mathbb N$ odd 
			\begin{align*}\label{eq:them24:bound}
			\| 	 \Psi(t)  - \Psi^{(n)}_{	\eff}(t)   \|_{\mathscr H}  
			& 	\leq  
\, C^n \sqrt{n! }  \, 
			\bigg(
			\frac{	\Vu  }{\mu^{	\frac{1}{2}	}}
			+ \frac{\Vu^n}{\mu^{	\frac{1}{2}}  (\mu m )^{\frac{n-1}{2} } } 
			+ \frac{|t|}{\mu^{	\frac{1}{2}}}
						\frac{ \Vu^{n+1} }{    (	  \mu m	)^{  \frac{n-1}{2}}			 }
		\min\bigg\{ 
			1, 	 	\frac{	\Vu 	}{\mu^{	\frac{1}{2} } m }  \bigg\} 	\bigg)  
		\end{align*} 
	for all $t \in \RE  $, $m \ge \mu^{-1}$ and all $\mu>0$ large enough.
	\end{theorem}
	
	\begin{remark} 
		In Theorem \ref{thm 2} 
	 the mass value $m = \mu^{-1 }$ is \textit{critical}, in the sense that the dependence of our estimates does not improve when invoking higher orders in $n$, 
		relative to  the massless case. That is,   the $n=1$ iteration from   Theorem \ref{thm 1}.  
	\end{remark}
	
	\begin{remark} Additional remarks are  in order. For concreteness, we assume that $V$ is given by \eqref{eq:Nelson:form:factor} so  that $\n V\n_\mu^2 = 
		O (  \log (\mu ) )  $.

\begin{enumerate}[leftmargin=*]
		\item	 Consider $n=1$. Then, for   $ \mu^{-1} \leq m  \le   (\log \mu)^{1/2} \mu^{-1/2} $ 
		the statement coincides with the massless case treated in Theorem \ref{thm 1}. 
	For $m\gg   (\log \mu)^{1/2} \mu^{-1/2}    $, the bound improves compared to Theorem \ref{thm 1}, as the approximation now holds for all  
	$$	|t | \ll  \frac{m \mu }{	(\log \mu)^{3/2}		}	  \ . $$
 
		\vspace{1mm}
 
\item For $n\ge 3$ and above the critical mass threshold, the approximation extends to substantially longer times. For instance, if $m \ge 1 $, convergence holds for all 
	$$|t| \ll   
	\frac{\mu^{	 \frac{n+1}{2}	}}{	 (\log \mu)^{\frac{n+2}{2} }		} \ . 
$$  
 
\vspace{1mm}

\item The reason for considering only $n\in \mathbb N$ odd is that $g_n = g_{n-1}$ for all $n\ge 2 $ even. This follows from the simple fact that 
		$	 \langle \Omega , b_{k_{j}}^{\# } \ldots  b_{k_1}^{\# } \Omega\rangle  = 0 $
		for all $j\in \mathbb N$ odd.

\vspace{1mm} 
\item Similarly as in the massless case,  it is possible to introduce 
  polynomial generators  $\h_{g_n}^{(N)}$ and to prove a statement analogous to Corollary \ref{cor:explicit:g}. The polynomial generators result from  the truncation  of a power series
  that   is obtained via expanding  the denominators  in \eqref{def gn} around		$\sum_{j} \mu \omega(k_j)$. 
  Thanks to the mass term, this expansion is in fact less singular that in the massless situation. 
Since this procedure is relatively straightforward, we omit the details.
\end{enumerate}
	\end{remark}

	\subsection{Comparison with previous results} 
	\label{subsection:comparison} Let us compare our analysis with some related results from the literature. 
	
\vspace{2mm}

\noindent (1) Davies \cite{Davies} considers heavy tracer particles weakly coupled to a scalar boson field. He utilizes the Hamiltonian $H^{(a)} = \sum_{i=1}^a (-\Delta_{X_i}) + T + \sum_{i=1}^a \mathbb V_{X_i} $ on the Hilbert space $\otimes^a L^2(\RE^3) \otimes \mathscr F$, with $T$ and $\mathbb V_{X_i}$ as defined in \eqref{V term}. For $a=1$, this Hamiltonian coincides with our choice of $H$. He establishes \cite[Theorem 2.2]{Davies} that for $t\in \RE$ and all $\varphi^{(a)} \in \otimes^a L^2(\RE^3)$,
\begin{equation}
\textstyle 
	\lim\limits_{\mu \to \infty }	e^{- i t  H^{(a)  } } \vp^{(a)}  \otimes \Omega 
	= 
	e^{- i t  \mathfrak h^{(a)}  }  \vp^{(a)} \otimes \Omega  \label{eq:Davies},
\end{equation}
where $\mathfrak h^{(a)} = \sum_{i=1}^a (-\Delta_{X_i}) +  K^{(a)} $ is an operator on $\otimes^a L^2(\RE^3)$, and $K^{(a)}$ describes an effective pair potential between the tracer particles, given by
\begin{align}\label{K:a}
K^{(a)} = \sum_{1\le i, j \le a }
 \int 
\cos(k(X_i-X_j) )
 \frac{|V(k)|^2}{\omega(k)} 
  \d k.
\end{align}
Note in particular that the effective Hamiltonian in \eqref{eq:Davies} does not contain any corrections of the free dispersion relation. This is due to the restriction to times $t=O(1)$ as $\mu \to \infty$.

For a single tracer particle, we reproduce \eqref{eq:Davies} with an explicit rate of convergence by Corollary \ref{cor:explicit:g}, when applied for $n=0$. To see this, note that $\alpha_0(\mu) = K^{(1)} + O(\mu^{-1})$ as $\mu \to \infty$. In his proof, Davies shows convergence of the resolvent, in the sense that
\begin{align}\label{eq:resolvent}
\textnormal{s}- \lim_{\mu\to \infty} \big( H^{(a)}+ z \big)^{-1}  = \big( \mathfrak h^{(a)}+z \big)^{-1} \otimes \P_\Omega,\quad z \in \mathbb C \setminus \mathbb R,
\end{align}
where $\P_\Omega$ is the projection in $\F$ onto the vacuum state. This is then used to obtain \eqref{eq:Davies}. Note that \eqref{eq:resolvent} does not imply strong convergence in the resolvent sense, as the limit operator is not the resolvent of an operator (for an abstract discussion of such limits, we refer to \cite[Sect. II]{Arai}). Our approach is based on a different strategy. We directly compare the unitary evolutions via repeated application of a suitable integration-by-parts formula (see Section \ref{section prelim}). Unlike \eqref{eq:resolvent}, this method directly provides quantitative error estimates. More importantly, it allows us to obtain an effective operator defined through a nonlinear equation as in \eqref{def g}.

Let us also mention that our analysis can be readily extended to reproduce Davies' result for $a\ge 2$ as well. However, our primary focus lies in extending the approximation to longer time scales, particularly those requiring non-trivial modifications of the dispersion relation. For transparency, we limit our investigation to the case where $a=1$, postponing the analysis of more than one tracer particle for future study. In such cases, we would expect momentum-dependent corrections to the effective pair potential $K^{(a)}$.

\vspace{2mm} 

\noindent (2) In \cite{Hiroshima,Hiroshima2} Hiroshima studies the same model as Davies, but with a $\mu$-dependent UV cutoff in \eqref{eq:Nelson:form:factor}. He establishes \eqref{eq:resolvent} in the simultaneous limit $\Lambda(\mu)\to \infty$ as $\mu \to \infty$, which requires the subtraction of a diverging energy. In this limit, the operator $K^{(a)}$ in \eqref{K:a} describes a Coulomb $(m=0)$ resp. Yukawa ($m>0$) pair potential. 

In the present work, we opt to maintain a fixed UV cutoff in \eqref{eq:Nelson:form:factor}. However, we expect that our results can be extrapolated to this coupled limit as well. In particular, one can verify that $H-g(0)$, 
 $g(p)-g(0)$, and consequently $(e^{-itH} - e^{-it\h_g})\Psi_0$, have appropriate limits as $\Lambda \to \infty$. Nevertheless, since our error estimates are not uniform in the cutoff -- note that $\| V\|_{L^2}\sim \Lambda $ and $\| |k|^{-1/2} V \|_{L^2} \sim \Lambda^{1/2}$ -- this observation alone is not sufficient for a direct extrapolation. We expect that refinement error bounds can be attained through the use of Gross' unitary dressing transformation. However, to keep the focus on the conceptual ideas, we refrain from introducing additional technical details. Finally, it is worth mentioning that we are unaware of any work that establishes \eqref{eq:Davies} for the renormalized Nelson model.

\vspace{2mm}

\noindent (3) The works by Davies and Hiroshima have been revisited and extended to other particle-field models as well as different scalings. See for instance \cite{Arai,Takaesu,Takaesu2,Takaesu3,
Hiroshima3,HiroshimaSasaki}.
 
\vspace{2mm}

\noindent (4) Teufel \cite{Stefan,Stefan2} and Tenuta and Teufel \cite{TenutaTeufel} consider the Neslon model in a slightly different scaling, namely they study the dynamics for large $\mu$ generated by
\begin{equation} \label{eq:Teufel:Ham}
\mathfrak H
  \   =  \ 
 \,  \frac{ - \Delta_X }{\mu}
 \, +     \, 
\int_{\RE^{3}} \omega (k) b_k^* b_k \d k 
 \ +   \,  
\int_{\RE^{3}}
	V(k) 
\big( 
e^{i k \cdot X}  
 \, b_k + 
e^{ - i k \cdot X}  
\, b_{k}^* 
\big)   \d k  \ ;
	\end{equation}
see e.g. \cite[Eq. (10)]{TenutaTeufel} for $\varepsilon =\mu^{-1/2}$. Comparing our Hamiltonian $H$ to $\mu \mathfrak H$, the latter has an additional factor $\mu^{1/2}$ in front of the third term, which makes the interaction comparable to the field energy $T$. For a heuristic comparison of the two scalings, see \cite[Section 6.3]{Stefan}. Conceptually, the scaling in \eqref{eq:Teufel:Ham} is similar to the one observed in the Born--Oppenheimer theory for electrons coupled to heavy nuclei. In fact, both can be understood through the general framework of
\textit{adiabatic perturbation theory}, and we refer to the book of Teufel \cite{Stefan} for an in-depth exposition of the subject. In this framework, Teufel \cite{Stefan2} and Tenuta and Teufel \cite{TenutaTeufel} obtained a norm approximation for the dynamics generated by \eqref{eq:Teufel:Ham}; see e.g.
\cite[Theorem 6.10]{Stefan}. More concretely, the authors consider
small electron velocity, and 
the wave function approximation consists of   
wave packets  $\int^\oplus_{\RE^3 } \vp(t,X) \Omega(X) \d X  $, 
where $\Omega(X)$
is the ground state of the $X$-dependent Fock space operator
$ H(X) =  
\int \omega (k) b_k^* b_k \d k   + a(  e^{ikX} V) + a^*(e^{ikX }  V  ) $, and 
$\vp(t,X)$ is driven by an effective Hamiltonian. Depending on the time scale, the effective Hamiltonian contains an effective pair interaction (including \eqref{K:a}, but also the momentum-dependent Darwin term) if one considers more than one tracer particle, and corrections to the mass of the particle(s). Contrary to our results, the effective states are \textit{not} vacuous, but rather represent dressed electron states.

\vspace{2mm}

\noindent (5) The emergence of an effective dispersion relation for a tracer particle coupled to a quantum field has also been explored in other contexts. Specifically, in \cite{Bach}, Bach, Chen, Faupin, Fr\"ohlich and Sigal, and previously in \cite{TeufelSpohn2}, Spohn and Teufel investigated the dynamics of an electron in a slowly varying external potential within the framework of non-relativistic QED. Their results show that the dynamics of such systems can be described in terms of dressed electron states whose evolution is governed by an effective dispersion relation $E(p)$. This dispersion is characterized, as usual for translation invariant particle-field models, as the infimum of the energy at fixed total momentum $p\in \RE^3$. 
Notably, while $E(p)$ plays a similar role as our $\h_g(p)$, 
our description of the generator is more explicit. In particular, we describe the effective dispersion as an explicit function of $|p|$ (up to small errors; see e.g. \eqref{eq:approx:sol:g}) and do not refer to the fiber decomposition of the Hamiltonian.

	\section{The integration-by-parts formula}
	\label{section prelim}
	The main goal of this section is to introduce 
	an integration-by-parts formula
	that will be the heart of our analysis. The formula  provides an expansion scheme for the difference between the original and the effective dynamics, and will be used in the proofs of Theorems \ref{thm 1} and \ref{thm 2}. 
	
	\vspace{2mm}
		
The heuristics behind the integration-by-parts formula are as follows: Upon invoking Duhamel's formula, see Eq. \eqref{eq:Duhamel} below, the propagator $e^{isH}$ exhibits rapid oscillations as $\mu \to \infty$, when acting on states in $L^2_X \otimes (\text{Span} (\Omega))^\perp$. This behavior arises from the presence of the large value of the operator $T =\mu \int \mathrm{d} k \omega(k) b_k^* b_k $ when applied to such states. Effectively, it suppresses the value of the integral coming from contributions in $L^2_X \otimes (\text{Span} (\Omega))^\perp$. On the other hand, contributions in the orthogonal complement $L^2_X\otimes \text{Span}(\Omega)$ do not manifest these rapid oscillations since $T\Omega = 0$. Later, we will demonstrate how these terms add up to the effective generator $g(p)$.

		\vspace{2mm}

	Throughout  this section, we denote 
	$\Psi_{\eff}(t) = e^{-it H_{\eff}} \Psi_0$ where $\Psi_0  = 	\vp \otimes \Omega \in \mathscr{H}$ with $\varphi \in D(p^2)$.
	The effective Hamiltonian is 
	\begin{equation}
		H_{\eff} = \h_{g }(p) \otimes \1 
	\end{equation}
	where  the  generator $g(p)$   is  not fixed but arbitrary, 
	and for notational convenience we drop the subscript, i.e. we write 
	$\h \equiv \h_g = p^2 -  g(p) $.

	\vspace{2mm}

	Let us start with the following   calculation
	for the difference between the original and the effective dynamics.
	Namely
	\begin{align}
		\nonumber 
	\Psi(t)		- 	\Psi_{\eff} (t)  & = 
		e^{ -i t H }
		\Big(
	 \1 - 	 e^{i t H } e^{- i t H_\eff}  
		\Big) \Psi_0	\\ 
		\nonumber 
		& = 
		e^{ -i t H }
		\frac{1}{i}
		\int_0^t 
		e^{i s H }
		\Big(
 H -  		H_{\eff}	  
		\Big)
		e^{- i s H_\eff}
		\Psi_0
		\d s  \\ 
		& = 
		e^{ -i t H }
		\frac{1}{i}
		\int_0^t 
		e^{i s H }
		\Big(
		 \mathbb V + g(p) 
		\Big)
		e^{- i s H_\eff}
		\Psi_0 \d s  \  \label{eq:Duhamel}
	\end{align}
	where we used $T\Omega = 0$.
	Let us now note 
	that thanks to
	$P \Omega = 0$
	we have $H_\eff \Psi_0 = \h (p+ P) \Psi_0$.
	Hence, since $[\V, p+P]= [g(p), p+P]=0 $, we obtain the identity
	\begin{equation}
		\Psi(t)	 - 			\Psi_{\eff} (t)   
		= 
		e^{ -i t H }
		\frac{1}{i}
		\int_0^t 
		e^{i s H }
		e^{- i s \h( p + P) } \d s  \
		\Big(
  \V  + g(p) 
		\Big)
		\Psi_0  \  . 
	\end{equation}
	
	The last expression motivates the following definition.
	
	\begin{definition}
		For all $t \in \RE$ we define the operator 
		\begin{equation}
			\I (t) : = 	\frac{1}{i}
			\int_0^t 
			e^{i s H }
			e^{ - i s  \h (p+P)	} \d s   \ . 
		\end{equation}
	\end{definition}

	\begin{remark}[Dynamics difference]
		\label{remark difference}
		Clearly, for all $ t    \in \RE  $, we can now write
		\begin{equation}
			\label{psi 1}
	\Psi(t)	  - 				\Psi_{\eff} (t)  
			= 
			e^{ - i t H } \I (t)  
			\Big(
		 \V + g(p) 
			\Big)
			\Psi_0    \ . 
		\end{equation}
	\end{remark}

	Note that 
		the first  term  $\V \Psi_0 $ is 
a  state that contains one boson. 
We are now interested in studying 
 the action of the operator $\I(t)$  on states
 that are   orthogonal to the vacuum, as they  lead
  to rapidly  oscillating  phases. 
 To this end, we introduce  
 \begin{equation}
\label{def:projections}
 	\P_\Omega = \1 \otimes  \ket{\Omega}	\bra{\Omega} \qquad
 	\t{and}
 	\qquad 
 	\Q_\Omega 
 	 = 
 	 \1   -\P_\Omega \ . 
 \end{equation}

	It will be convenient to define the following two auxiliary operators.
	The first one we call the \textit{resolvent}.
	The second one we call the \textit{boundary term}. 
	
	\begin{definition}[Auxiliary operators]
		\label{def R and B}
		We define  as operators on $\mathscr{H}  =  L_X^2\otimes \F$
		\begin{enumerate}
			\item The resolvent      
			\begin{equation}
				\R : = 	
				\Q_\Omega 
				\Big(
				\h(p+P) - p^2 - T 
				\Big)^{-1 } \Q_\Omega
			\end{equation}
			\item The boundary term, for all $t \in \RE $
			\begin{equation}
				\B (t)  : =	\Big(
				e^{i t H } e^{ - i t  \h(p+ P )} - \1 
				\Big) \ . 
			\end{equation}
		\end{enumerate}
	\end{definition}

	\begin{remark}[Domain of $\R$] 
The resolvent operator $\R$ is an 
 unbounded 	  operator on $\H$.
 Here, 
 we define its domain   as
 $D (\R) := \{	 \Psi \in \H : 
   \t{s--}\lim_{\epsilon\downarrow 0 }  \R(\epsilon )	\Psi  \t{ exists in } \H 		\}$, 
    where  
\begin{align}
\R (\epsilon ) := 
\Q_\Omega 
\Big(
				\h(p+P) - p^2 - T - i \epsilon
				\Big)^{-1 } \Q_\Omega  \ , 
				 \quad \epsilon >0.
\end{align}
We then define the strong limit
$\R \Psi : = \lim_{\ve\downarrow 0 } \R(\epsilon )\Psi$. 
In practice, we shall always apply $\R$ to states $\Psi \in \mathscr H$ that are evidently
 in the domain $D(\R)$--see the  estimates contained 
in Sections \ref{section massless} and \ref{section resolvent estim}, respectively.
In order to keep the exposition simple,  
we do not refer anymore to the unbounded nature of $\R$ outside of this section. 
\end{remark}

	Let us now relate $\I (t)$, $\R$ and $\B (t)$. 
	The following proposition 
	contains
	the core idea of our proof. 
	\begin{proposition}[Integration-by-parts formula]
		\label{prop integration by parts}
For all $\Psi  \in  D(\R) \cap D(\V\R)$, we have 	  
		\begin{equation}
			\label{I(t)}
			\I(t)  {  \Q_\Omega  }  \Psi = \Big( \B(t) \R 		+		\I(t) \V \, \R \Big) \Psi  \ . 
		\end{equation}
	\end{proposition}

\begin{remark}
	Let us comment on the above proposition.
	
		\begin{enumerate}[leftmargin=*]
		\item 
		In practice, 	we only apply the integration-by-parts formula 
		to  states $\Psi \in \Q_\Omega \H$.
		Thus, unless confusion arises,  we will often omit the projection  $\Q_\Omega$.
		
		\vspace{1mm}
		
		
		\vspace{1mm}
		
		\item  Proposition \ref{prop integration by parts} is partially inspired by previous works \cite{Jeblick,
			Jeblick2,Mitrouskas2,Mitrouskas,
			Leopold}, where similar integration-by-parts formulas were used in different contexts. The common feature of the considered models is that a slow system is coupled to a fast system, and that the evolution of the slow system is governed by an effective generator. However, we want to emphasize that our expansion is novel in that it yields an effective generator that solves a \textit{nonlinear} equation; see e.g. \eqref{def g}.
		
	\end{enumerate}

\end{remark}

	\begin{proof}

	In what follows, we understand  
the operator identities to hold on $\Q_\Omega \H$.
		The main idea is to integrate by parts in a convenient way.
		To this end, we first compare
		with the free kinetic energy
		\begin{align}
			\I (t) 
			& = 
			\frac{1}{i}
			\int_0^t 
			e^{i s H }
			e^{ - i s  \h (p+P)	} \d s    \\ 
			& = 
			\frac{1}{i}
			\int_0^t 
			e^{i s H } e^{   - is ( p^2 + T )}
			e^{ i s (p^2 + T )}
			e^{ - i s  \h (p+P)	} \d s     \ . 
		\end{align}	
		Note that  $[ \h(p+ P) , p^2 + T]=0	$. Therefore, we can write 
		\begin{align}
			\I (t) 
			& = 
			\frac{1}{i}
			\int_0^t 
			e^{i s H } e^{  -  is ( p^2 + T )}
			e^{ - i s  \big(  \h (p+P)	 - p^2 - T \big) 				} 
			\d s     \ . 
		\end{align}
		Next, we use  
		\begin{align}
			e^{ - i s  \big(  \h (p+P)	 - p^2 - T \big) 				} 
			=  \frac{ \d }{\d s }  
			e^{ - i s  \big(  \h (p+P)	 - p^2 - T \big) 				}  
			i  \R
			\ . 
		\end{align}
		and integrate by parts
		\begin{align}
			\nonumber 
			\I(t) 
			& =  
			\frac{1}{i } 
			e^{i s H } e^{  -  is ( p^2 + T )}
			e^{ - i s  \big(  \h (p+P)	 - p^2 - T \big) 				} 
			i  \R
			\bigg{|}_{s= 0}^{t}	\\
			& \quad 	- 
			\label{eq I}
			\frac{1}{i } 
			\int_0^t 
			e^{i s H }
			\frac{1}{i}
			\Big(
			p^2 + T  - H 
			\Big)
			e^{  -  is ( p^2 + T )}
			e^{ - i s  \big(  \h (p+P)	 - p^2 - T \big) 				} 
			i  \R 
			\d s  \ . 
		\end{align}
		The first term in \eqref{eq I} corresponds to
		the boundary term $\B(t)$. Namely, 
		using again $[ \h(p+ P) , p^2 + T]=0	$
		\begin{equation}
			\frac{1}{i } 
			e^{i s H } e^{  -  is ( p^2 + T )}
			e^{ - i s  \big(  \h (p+P)	 - p^2 - T \big) 				} 
			i  \R
			\bigg{|}_{s= 0}^{t}
			= \B(t) \R  \ . 
		\end{equation}
		For the second term in  \eqref{eq I}
		we observe that $p^2+T -H = - \V$. 
		Using again 
		$[ \h(p+ P) , p^2 + T]=0	$
		and $[ \V , p + P]=0 $
		we find 
		\begin{align}
			\nonumber 
			- 	\frac{1}{i } 
			\int_0^t 
			e^{i s H }
			\frac{1}{i}
			\Big(
			p^2 + T  - H 
			\Big)
			e^{  -  is ( p^2 + T )}
			& 	e^{ - i s  \big(  \h (p+P)	 - p^2 - T \big) 				} 
			i  \R 
			\d s   \\ 
			\nonumber 
			& =   
			\frac{1}{i } 
			\int_0^t 
			e^{i s H }
			\frac{1}{i}
			\V  
			e^{  -  is ( p^2 + T )}
			e^{ - i s  \big(  \h (p+P)	 - p^2 - T \big) 				} 
			i \R 
			\d s   \\ 
			\nonumber 
			& =   
			\frac{1}{i } 
			\int_0^t 
			e^{i s H }
			\V  
			e^{ - i s    \h (p+P)	 			} 
			\R 
			\d s   \\ 
			& =    
			\frac{1}{i } 
			\int_0^t 
			e^{i s H }
			e^{ - i s    \h (p+P)	 			} 
			\d s    
			\V  	\R  
			=   \I(t) \V \R \ . 
		\end{align}
		This finishes the proof of the proposition. 
	\end{proof}

	\section{Effective dynamics for massless fields}
	\label{section massless}
	In this section, we apply the integration-by-parts formula 
	\begin{equation}
		\label{int by part formula}					\I(t) = \B(t) \R 		+		\I(t) \V \, \R 
	\end{equation}
	provided by Proposition \ref{prop integration by parts}, 
	in order to prove Theorem \ref{thm 1}. 
	In what follows, we always assume that $\Psi_0$  satisfies Condition \ref{condition 1} relative to  some fixed $\ve $, 
	and that $   V $ satisfies Condition \ref{condition 2}.  
	The field is either massless or massive, 
	but the results in this section are mostly relevant to the massless  case. 
	
	\vspace{1mm}
	We choose  the function  $g (p)$  according to Definition \ref{def g1}
	and  unless confusion arises  we drop the
	following subscript $\h = \h_g$. 
	In addition,  we use the decomposition   $\V  = \V^+ + \V^-$
	\begin{align}
		\label{decomposition}
		\V^+ 
		\equiv
		\mu^{1/2}
		\int_{\RE^3}  V(k) e^{ -  i k \cdot X}b_{k}^* \d k    
		\qquad \t{and}\qquad 
		\V^- 
		\equiv 
				\mu^{1/2}
		\int_{\RE^3}  V(k) e^{ i k \cdot X}b_{k} \d k   
	\end{align}
	in terms of creation and annihilation operators.

	\vspace{2mm}

	Let us recall that the difference $\Psi_{\eff}(t)  -  \Psi(t)$
	was written in terms of $\V$ and $g(p)$ in Remark \ref{remark difference}.
	In what follows, we      integrate-by-parts  the $\V$ term    
	using \eqref{int by part formula}
	and $\V \Psi_0 \in \Q_\Omega \H$, 
where
 $\P_\Omega$ and 
 $\Q_\Omega = \1  - \P_\Omega$ were the projections  introduced in \eqref{def:projections}. 
We find that 
	\begin{align}
		\nonumber 
		e^{i tH } \Big( 
		  \Psi(t) - 	\Psi_{\eff}(t)   
			\Big) 
		& = 
		\I(t) 
		\Big(
		\V  + g(p) 
		\Big)
		\Psi_0      \\ 
		\nonumber 
		& = 
		\I(t) 
		\Big(
		\V \R \V   + g(p) 
		\Big)
		\Psi_0   
		+ 
		\B (t)  \R \V \Psi_0  \\ 
		& = 
		\I(t) 
		\Big(
		\P_\Omega \V \R \V   +  g(p) 
		\Big)
		\Psi_0   
		+ 
		\I(t) 
		\Q_\Omega \V \R \V \Psi_0 
		+ 
		\B (t)  \R \V \Psi_0 \ . 
		\label{step1}
	\end{align}

	The next step is to realize that  our choice of  $g(p)$ 
	forces the first term on the right side to vanish.
	The other two terms then need to be estimated. 
	Let us now calculate the projection onto the vacuum. 
	\begin{lemma}\label{lemma P}
		The following identity holds 
		\begin{equation}
			\label{psi 3}
			\P_\Omega \V \R \V \Psi_0 
		 \, 	=		 \, 
			 - \mu 
			\int   | V (k)|^2 \Big( 
			  \mu\, \omega(k)
			  + 
			  (p-k)^2
			   - \h(p) 
			\Big)^{-1 }   \d k\otimes \1 \  \Psi_0 \ . 
		\end{equation}
	\end{lemma}

	An immediate consequence of 
	Lemma \ref{lemma P}    is  
	that  thanks to the  choice  of $g  (p)   $ given by Def. \ref{def g1}, it follows from   
	\eqref{step1}    that 
	\begin{equation}
		\label{eq psi 1}
		\|	 \Psi(t)		-		\Psi_\eff (t)	 \|_{\mathscr{H}} 
		\leq  
		\|    \B(t) \R \V \Psi_0		\|
		+ 
		\|	 \I(t)	\Q_\Omega \V \R\V \Psi_0 	\|  \ . 
	\end{equation}
	We dedicate the rest of this section to the proof of Lemma \ref{lemma P} 
and 
	to provide estimates on the remainder terms given by the right hand side of \eqref{eq psi 1}. 
	These estimates are given in Proposition \ref{prop B} and \ref{prop Q}. 
	We informally refer to these two terms as the boundary term, and the projection term, respectively.

	\begin{proof}[Proof of Lemma \ref{lemma P}]
		We decompose $  \V$
		and consider only the combination
		of creation and annihilation operators
		that give a non-zero contribution.
		Namely
		\begin{equation}
			\label{p eq 1}
			\P_\Omega \V \R \V \Psi_0 
			=
			\P_\Omega 
			\int  
 \mu 
			V(k_1)   V( k_2 ) 
			\Big( 			e^{ i k_2 X} 
			b_{k_2} 
			\R  
			e^{ - i  k_1  X} 
			b_{k_1 }^* \big) 
			\Psi_0 	\,		\d k_1
			\d k_2  \   . 
		\end{equation}
		Next, we compute using standard commutation relations that 
		\begin{align}  
			\R   
			e^{ - i  k  X} 
			b_{k}^* 
			& 	  = 
			e^{ - i  k  X} 
			b_{ k}^*  
			\Big(
			\h(p  +P) - (p - k )^2 - T  -  \mu  \,  \omega(k) 
			\Big)^{-1 } \ . 
			\label{p eq 2}
		\end{align}
		In particular,  if we evaluate the above operator on $\Psi_0 = \vp \otimes \Omega$ we find 
		\begin{equation}
			\label{p eq 3}
			\R   
			e^{ - i  k  X} 
			b_{k}^*  \Psi_0 
			=  
			e^{ - i  k  X} 
			b_{ k}^*  
			\Big(
			\h(p  ) - (p - k )^2    - \mu \,  \omega(k) 
			\Big)^{-1 } \Psi_0 \  , 
		\end{equation}
		in view of  $T\Omega = P\Omega=0$. 
		 The CCR now imply that for all $k_1, k_2  \in \RE^3$
		\begin{equation}
			\label{p eq 4}
			\P_\Omega b_{k_2} b_{k_1}^*  
			=  \delta(k_1 + k_2 ) \P_\Omega 
		\end{equation}
		since $\P_\Omega b_{k }^* = 0$. 
		The proof is finished
		once we put together \eqref{p eq 1}, \eqref{p eq 3} and \eqref{p eq 4}. 
	\end{proof}

	Before we turn to the analysis of the remainder estimates,
	let us record here the proof of Theorem \ref{thm 1}.
	
	\begin{proof}[{Proof of Theorem \ref{thm 1}}]
		Starting from the   estimate \eqref{eq psi 1} 
		we use the results from 
		Proposition \ref{prop B} to estimate the boundary term,
		and the result from Proposition \ref{prop Q}
		to estimate the projection term. 
	\end{proof}

	\subsection{Error analysis}  
	Recall that the initial datum $\Psi_0 = \vp \otimes \Omega$
	satisfies Condition \ref{condition 1} with respect to $\ve\in(0,1/2)$, 
	and that $\omega(k) \ge |k|$.
	We also remind the reader of the notation 
	$\omega_\mu (k)   = \omega (k)  + \mu^{-1 }$. 
	
	\begin{proposition}[Boundary term]\label{prop B}
		There exists a constant  $C>0$	 such that  
		\begin{equation}
			\|	 \B(t) \R \V \Psi_0	\|_{\mathscr H}
			\leq 
			\frac{ C }{\mu^{1/2}}
			\|   \omega_\mu^{-1 }  V	\|_{L^2} 
		\end{equation}  
for all $t \in\RE $ and $\mu>0$ large enough. 
	\end{proposition}
	
	\begin{remark}
		We 
		record here     the usual estimates      
		for creation-
		and annihilation operators,  extended to $ \mathscr H = L_X^2 \otimes \F $.
		The following is enough for our purposes:  for $ n \in  \mathbb N $
		and 
		$f  \in L^2(\RE^{3n} ; 	 L^2_X	)$ 
		we have that 
		\begin{equation}
			\label{number estimate}
			\bigg\|
			\int_{ \RE^{3n}}    f(k_1, \cdots , k_n ) \otimes 
			b_{k_n}^* \cdots 	b_{k_1}^* \Omega  \  
			\d k_1 	 \cdots \d k_n 
			\, 		\bigg\|_{\mathscr H} 
			\leq  \sqrt{ n!  }		 \|	f 	\|_{L^2 (	L^2_X	)	}  \ . 
		\end{equation} 
	\end{remark}
	
		\begin{remark}\label{remark bounds g}
		In Section \ref{section g} we 
		describe in more detail  the generator $g(p)$. 
		In the present section, we will only need the following bounds,  valid for $\mu$ 
		large enough: 
		\begin{equation}
			g(p)  > 0 
			\qquad 
			\t{and}
			\qquad 
C_1 
			\leq 
			|g(p)|
			\leq 
C_2 
		\end{equation}
		for $|p|\le \ve \mu$ and for some constants $C_1$ and $C_2$. 
		See Lemma \ref{lemma g} for more details. 
	\end{remark}

	\begin{proof}[Proof of Proposition \ref{prop B}]
		First, we note that 
		$		\|	 \B(t) \R \V \Psi_0	\| \leq 2 		\|	  \R \V^+ \Psi_0	\| $.
		Further, 
		we use the decomposition \eqref{decomposition}
		to write $\V \Psi_0 = \V^+ \Psi_0	$.
		Then,  thanks to   \eqref{p eq 3} we find 
		\begin{equation}
			\label{RV+}
			\R \V^+ \Psi_0
			= 
			\mu^{1/2 }
			\int_{\RE^{3}} 
			V(k) e^{- i k X} 
			\Big(
			\h(p ) - (p - k )^2   -  \mu \omega(k) 
			\Big)^{-1 } \vp \otimes b_k^* \Omega   \ \d k  \ . 
		\end{equation} 
		Next,   the  number estimates 
		\eqref{number estimate}
		imply that in the  momentum representation 
		\begin{align}
			\label{eq1}
			\|		 \R \V^+ \Psi_0	\|_{\mathscr{H}}^2
			& 	\leq  
			\int_{\RE^{6}} 
			\frac{  \mu \, 	|	   V (k)|^2  \, 		|  \hat \vp(p)	|^2 \, 	\d p \d k  	}{	
				|  \h(p ) - (p - k )^2   -  \mu \omega(k) 	|^2	} 	 \   . 
		\end{align}
		The next step is to find an appropriate lower 
		bound of the   denominator in the above integrand.
		To this end, we use the lower bound 
		$|g(p)| \geq C_1 $ from Remark \ref{remark bounds g}, 
		the fact that $|p| \leq \ve \mu $
		and $\omega(k) \geq \mu |k|$ to find that 
		\begin{align}
			\nonumber
			|  \h(p ) - (p - k )^2   -  \mu \omega(k) 	|
			& \   \geq  \ 
		\mu 	\omega(k)	+  2 p \cdot k + k^2 +  g(p)   	\\ 
			\nonumber
			&  \ 	\geq  \   
			(1 - 2 \ve) \mu \omega(k) 
			+ C_{1 } 	\\ 
			& \ 		\geq  \ 
			C \mu 
			\big(	 \omega(k) +  \mu^{-1 }		\big) 
			\label{eq11}
		\end{align} 
		where $C = \min(	 1- 2 \ve , C_1	)$. 
		The proof   is finished once
		we combine \eqref{eq1}
		and \eqref{eq11}. 
	\end{proof}

	Next, we turn to   the following proposition, 
	in which we estimate the projection term.
	
	\begin{proposition}[Projection term]\label{prop Q}
		There exists a constant  $C>0$	 such that  
		\begin{equation}
			\|	 \I(t) 
			\Q_\Omega \V \R \V \Psi_0	\|_\H 
			\leq 
			\frac{C }{\mu }
			\big(	1 +    	\mu^{1/2}	 | 	t | 	\big)
			\|	  \omega_\mu^{-1 }   V	\|^2_{L^2}
		\end{equation}  
	for all $t \in \RE $ and $\mu>0$ large enough. 
	\end{proposition}

	\begin{remark}\label{remark obs}
		Before we turn to the proof,
		we make the following   observation.
		Let $ n \geq 1 $ and let   $\Phi \in \mathscr{H}$ be an $n$-particle state, i.e. $\calN \Phi = n \Phi$. 
		Then
		\begin{equation}
			\label{obs}
			\|	 \I(t)  \Phi 	\|_\H 
			\leq 
			\big(	\, 2 \, +  \, |t|  \,	 \mu^{1/2}	\,	 \|   V	\|_{L^2}  \, (n + 1)^{\frac{1}{2}}	 \, 	\big)  \|	 \R \Phi	\|_\H  	 \ . 
		\end{equation}
		Indeed, 	starting from \eqref{int by part formula} one uses that $\|	 \B(t)\| \leq 2 $
		as well as $\|	 \I(t)	\| \leq |t|$ in operator norm.
		Recall that $[\R,N]=0$. 
		Thus,     $\|	  \V \R \Phi 	\|  \leq   \| \V (\calN + 1)^{-1/2}	  (n+1)^{1/2} 	\R \Phi\|	$.
		Then, it suffices to use  the  standard estimate
		$ \| \V (\calN + 1)^{-1/2}	\| \leq	 \mu^{1/2}	 \|	   V\|_{L^2}$
		
	\end{remark}

	\begin{proof}[Proof of Proposition \ref{prop Q}]
		We use the observation \eqref{obs}
		as well as the decomposition \eqref{decomposition}
		of the interaction term to find 
		that for some constant $C>0$
		\begin{equation}
			\|		 \I(t)		\Q_\Omega \V \R \V \Psi_0		\|_\H 
			\ \leq  \ 
			C 
			\big(	1 +  \mu^{1/2} | t|  	\big)   
			\|	 \R \V^+ \R\V^+ \Psi_0	\|_\H   \ . 
		\end{equation}
		It suffices now to estimate the norm on the right hand side. 
		To this end,  we do a 
		two-fold application of the commutation relation \eqref{p eq 2}
		to find that
		in analogy with  \eqref{RV+} 
		\begin{align}
			\label{V++}
			\R \V^+ \R \V^+ \Psi_0 		
			= 
			\mu 
			\int_{\RE^{6 }}		   V(k_1)		  V(k_2)
			e^{-i (k_1 + k_2) X }
			R_p (k_1) R_p ( k_1, k_2)\vp \otimes b_{k_2}^* b_{k_1}^*	\Omega \d k_1 \d k_2  \ . 
		\end{align}
		Here,  $R_p (k_1)$  is the operator 
		that 
		appears  from the commutation of $\R$ and \textit{one} of the $b_k^*$ operators.
		The resolvent $R_p (k_1, k_2)$ appears due to the  commutation of $\R$ 
		and \textit{two} $b_k^*$ operators.
		They are defined   via spectral calculus for $p = -i \nabla_X$ 
		on the subspace $\1 (|p | \leq \ve \mu) L_X^2 $ via the formulae: 
		\begin{align}
\label{resolvent 1}
			R_p(k_1)  
			&   : = 
			\Big(
			\h_g(p) - (p - k_1)^2 -  \mu \omega (k_1) 
			\Big)^{-1 } \  ,  \\
\label{resolvent 2}
			\textstyle 
			R_p(k_1, k_2)
			&  :   =  
			\Big(
			\h_g(p) - (p - k_1 - k_2)^2 - \mu  \omega (k_1)  -\mu  \omega (k_2)
			\Big)^{-1 } \  , 
		\end{align}
and extended by zero to $L_X^2$. 
		In particular, for  $|p| \leq \ve \mu$
		we may replicate the lower bounds \eqref{eq11} 
		for the denominators  
		to find that the  following operator norm estimates hold, 
		for a constant $C>0$
		\begin{equation}
			\label{resolv estim 1}
			 \| R_p(k_1,k_2) \| 
			\leq 
				\frac{ C  }{\mu  \,  	 \omega_\mu(k_2) 		}   
			\qquad
			\t{and}
			\qquad 
			 \|	 R_p(k_1)	\| 
			\leq 
			\frac{C }{\mu   \,  \omega_\mu(k_1)}     \   . 
		\end{equation} 
		Finally,  we proceed analogously as we did in the proof of Proposition \ref{prop B}.
		That is, 
		we combine
		the number estimate \eqref{number estimate}
		and 
		the resolvent bounds  \eqref{resolv estim 1}
		to find  that  
		\begin{align}
			\nonumber 
			\|	   \R \V^+ \R \V^+ \Psi_0 			\|_{\mathscr{H}}^2
			&  \  	\leq  \ 
			2  \mu^2 
			\int_{\RE^{6}}
			|	   V(k_1)	|^2 \, 
			|	   V(k_2)		|^2 \, 
		\|	 R_p (k_1) R_p (k_1, k_2 ) 		\vp	\|^2_{L_X^2 }
		 \,  \d k_1 \d k_2 \ \\
			& \  \leq  \ 
			\frac{C }{\mu^2 }
			\|	  \omega_\mu^{-1 }   V	\|^4_{L^2}   \ . 
		\end{align}
		This finishes the proof. 
	\end{proof}

	\section{Analysis of the effective generator}
	\label{section g}
	In this section 
	we analyze    
	the effective  generator  $g(p)$
for  massless boson fields, with dispersion  $\omega (k) = |k|$.  That is, 
	the function 
	on $B_{\mu/2} \equiv \{ |p |< \frac{1}{2} \mu\}$
	satisfying 
	\begin{equation}
		\textstyle 
		g(p) = \mu \int_{\R^3 }		|V(k)|^2  \big( 	 \mu |k| + (p-k)^2 - p^2 + g(p )  \big)^{-1 } \d k 
		\label{g}
	\end{equation}
and first introduced in 
	Definition \ref{def g1}. 
We remind the reader that  in  Theorem \ref{thm 1} 
we proved the  validity  of the approximation with the effective Hamiltonian $\h_g = p^2  - g(p)$.

	\vspace{2mm}
	Our  analysis contains  two parts. 
	First, we   prove that 
	$g(p)$ can be approximately solved in terms of an    explicit  analytic  function of $p \in  B_{ \mu /2}$. 
	The error in the approximation is  uniform over compact sets of $B_{\mu/2}$. 
	Secondly, 
	based on this analytic approximation,   we 
	introduce a sequence of polynomial generators, corresponding   to the truncations of the   associated power series.
These polynomials 	 then induce a sequence of effective Hamiltonians $\h_g^{(N)}$. 
	Combined with Theorem \ref{thm 1}, 
	we  can study their validity as an approximation 
	of the original dynamics. This will lead to a proof of Corollary \ref{cor:explicit:g}. 
We then analyze the   optimality of our approach in Theorem \ref{thm:non-convergence}. 
	
	\vspace{2mm}
	
	Throughout this section, we   assume  that  $V$ satisfies Condition \ref{condition 2}.

	\subsection{Solving for the generator}
	The first step towards solving for $g(p)$
	is  the following result 
	that was already used in the last section.

	\begin{lemma}\label{lemma g}
	Fix   $\ve \in [ 0 , \frac{1}{2} )$. 		Then, 
		for all $ |  p |  \leq \ve \mu $ 
		\begin{equation}
			\bigg( 
			\frac{ 1  + \delta_{\ve,V} (\mu)}{ 1 + 2 \ve	}
			\bigg)
			\int_{\R^3} 
			\frac{|V(k)|^2}{|k| } \d k 
			\, 	\leq  \, 
			g(p) 
			\, 	\leq  \, 
			\bigg( 
			\frac{ 1   }{ 1 - 2 \ve	}
			\bigg) 
			\int_{\R^3} 
			\frac{|V(k)|^2}{|k| } \d k  \ ,
		\end{equation}
		where $\delta_{\ve,V}(\mu) \rightarrow 0 $ as $ \mu \rightarrow \infty$. 
	\end{lemma}

	\begin{proof}
		First, we prove the upper bound.
		Recall that $g(p)>0$ and
		therefore we can estimate  the denominator
		$| 	\h_g(p)	-	(p -  k )^2 - \mu \omega (k)|   \geq ( 1 - 2 \ve)  \mu |k| $.
		It suffices to plug this bound back in \eqref{g} and use  the triangle inequality.
		We write $C_2  \equiv (1 - 2\ve)^{-1} \int \d k |k|^{-1} |V(k)|^2$.

		\vspace{1mm}
		Now, we prove the lower bound. 
		To this end, let $ \theta  \in (0,1)$. 
		Then, we find  
		\begin{align}
			\nonumber
			g(p) & =  \mu
			\int_{\R^3 } \frac{   | V(k)|^2	 \d k 	}{ 		 \big(
				\mu \omega(k)  -  2 p\cdot k + k^2  + g(p) 
				\big)	 }  		\\ 
			& \geq \mu
			\int_{\R^3 } \frac{  | V(k)|^2	 \d k 	}{ 		 \big(
				(1 + 2 \ve ) \mu |k |+ k^2  +  C_2 
				\big)	 }  \notag \\ 
			\nonumber
			& \geq \mu
		 \int_{ 	 \{  \theta^{-1}  \mu^{-1} \leq |k| \leq \theta   \mu				\}	}
		\frac{ | V(k)|^2	 \d k 	}{ 		 \big(
				(1 + 2 \ve ) \mu |k |+ k^2  +  C_2 
				\big)	 }  	 \\
			& \geq 
			\frac{	   	  \int_{ 	 \{  \theta^{-1}  \mu^{-1} \leq |k| \leq \theta   \mu				\}	}	|k|^{-1 }|V(k)|^2	\d k	}{ \big(  1 + 2\ve + \theta  ( 1 + C_2)  	\big)   	}   \ . 
		\end{align}
		where in the last line we used $k^2 \leq  \theta  \mu |k|$
		and $1 \leq  \theta  \mu |k|$.  
		It suffices now to take $ \theta = \mu^{-1/2} $, 
		apply  the monotone convergence theorem, 
		and perform elementary manipulations. 
	\end{proof}

In our next result,  
	we describe the  generator $g(p)$
	in terms of an explicit  function of   $ p  \in B_{\mu/2} $, plus  a small error term. 
	To this end, we denote  $g_0 := g (0) > 0$. 
	We also introduce the notation for the denominator
	\begin{equation}
		D(k) 
		: =  \mu |k| + k^2 +  g_0  \ , \qquad k \in \R^3 \  , 
	\end{equation}
 which we use extensively in the rest of this section.

	\begin{proposition}\label{prop expansion g}
	Fix $\ve \in [ 0 , \frac{1}{2} )$. 
		Then, there is  a constant $C>0$
		such that 
		for all $|p|\leq \ve  \mu$ 
		\begin{equation}
			\label{g eq 2}
			\bigg|
			g(p)	
			-
			\frac{\mu}{2  |p|}
			\int_{\R^{3}} \frac{|V(k)|^2}{|k|}
			\tanh^{-1} \bigg(	 \frac{2 |p| |k |}{D(k)}		\bigg) \  \d k  
			\bigg|
			\leq 
			\frac{C }{\mu} \Vu^2  \  , 
		\end{equation} 
		for all  $\mu>0$ large enough. 
	\end{proposition}
	
	\begin{proof}
		Let us denote 
		$D (p,k)  \equiv   \mu |k| +k^2 - 2p\cdot k  +  g(p) $
		so that  $g(p) = \mu \int \d k  |V(k)|^2 / D(p,k)$. 
		In particular,    Lemma \ref{lemma g}  
		implies that   $C_1  \leq  g(p)  \leq C_2 $
		for a pair of constants $C_1$ and $C_2$. 
		Thus, we find  the following   lower bounds for the denominators
		\begin{equation}
			\label{D:lower}
			D(p,k)  \, \geq  \, 
			C_3 \,  (  \mu  |k| +  1 ) 
			\qquad \t{and}\qquad   D(k)	 \pm  2 p \cdot k 	 
			\, \geq \,  C_3 \,   (  \mu |k|	 + 1 	)
		\end{equation}
		for an appropriate constant $C_3>0$.

		\vspace{1mm}
		The first step towards \eqref{g eq 2}  is to  compare   $g(p)$  
		and its denominator $D(p,k)$
		with the simpler one $D(k) - 2p\cdot k $. 
		We find
		\begin{equation}
			\label{g eq 1}
			\frac{g(p)}{\mu}	 =  
			\int_{\R^3 }
			\frac{| V(k)	 	|^2		\d k 	}{  D(p,k)}  
			=
			\int_{\R^3 }
			\frac{| V(k)		|^2		\d k }{D(k) - 2 p \cdot k } 
			+
			\big(   g(p)  - g_0 	\big)  \int_{\R^{3}}	
			\frac{  |V(k)|^2\d k  }{	   D (p,k) \big(	 D(k) - 2 p \cdot k	\big)	}  .
		\end{equation}
		Thus, \eqref{g eq 1} and \eqref{D:lower} imply 
		\begin{align}
			\nonumber 
			\bigg| 	\frac{	g(p)}{\mu}	- 
			\int_{\R^3 }
			\frac{| V(k)		|^2		\d k }{D(k) - 2 p \cdot k } 
			\bigg| 
			\leq 
			\frac{			2 C_2}{C_3^2}
			\int_{\R^3} 
			\frac{|V(k)|^2 \, \d k }{	 \big(	  \mu |k|  + 1 	\big)^2					}	 
			=  
			\frac{C}{\mu^2} \Vu^2	 \ , 
		\end{align} 
		where 
		$
		C \equiv 
		2 C_2  / C_3^2
		$
		and we remind the reader that $\Vu$
		is given by \eqref{nu norm}. 
		This gives the estimate   \eqref{g eq 2}  in the statement of the proposition.
		
		\vspace{1mm}
		
		The second step  is   to analyze the integral on the left hand side of \eqref{g eq 2}. 
		To this end, we use that $V(k)$ and $D(k)$ are rotational symmetric, and thus we compute 
		\begin{align}
			\int_{\R^3 }
			\frac{| V(k)		|^2		\d k }{D(k) - 2 p \cdot k } 
			= 
			\int_{\R^3 }
			\frac{| V(k)		|^2	 D(k)	\d k }{( D(k) - 2 p \cdot k) (D(k) + 2 p \cdot k ) } 
			= 
			\int_{\R^3 }
			\frac{| V(k)		|^2	  	\d k }{D(k)  \big(   1 -     | 	 \frac{2 p \cdot k}{D(k)}	|^2 \big) 	 }  \ . 
		\end{align}
		In the last integral we can compute the 
		angular integration.
		Let us denote $\rho = |k| \in (0,\infty)$ and $\omega = k /|k| \in \mathbb S^2$.
		Without loss of generality, we may assume inside the integral that  $p \cdot  k = |p| k_3 $ and, therefore, 
		in spherical coordinates we compute
		\begin{align}
			\nonumber
			\int_{\R^{3}}	 	\frac{ |V(k)|^2}{D(k) }       \bigg(  	   \frac{1 }{ 1   -    | 	 \frac{2 p \cdot k}{D(k)}	|^2		}	  		\bigg)	\d k     
			\nonumber
			& = 
			\int_0^\infty  \d \rho   \, 
			\rho^2 		 \frac{|V(\rho)|^2		}{	D(\rho)	}
			\int_0^{2\pi} \d \vp 
			\int_0^\pi 
			\frac{\sin \theta \d \theta}{ 
				\big( 1 -  
				|	 \frac{|p|\rho \cos \theta}{D(\rho)}		|^2 	 \big) 		 }   	\\ 
			\nonumber
			& = 
			4 \pi 
			\int_0^\infty  \d \rho \, 
			\rho^2 		  \frac{|V(\rho)|^2		}{	D(\rho)	}
			\frac{			D(\rho) }{		2|p|\rho			}	
			\int_{ 0 }^{	 \frac{		2 |p|\rho}{D(\rho)}		
			} 
			\frac{\d x }{1 -   x^2 		 }   	\\[3mm]
			\nonumber
			& = 
			\frac{2\pi}{|p|}
			\int_0^\infty  \d \rho   \, 
			\rho	  |V(\rho)|^2
			\tanh^{-1}		\Big(    \frac{2 |p|\rho}{D(\rho ) }	  	\Big)   .
		\end{align}
		This finishes the proof. 
	\end{proof}
	
	\subsection{Polynomial approximations} 
	From Proposition \ref{prop expansion g}
	we conclude that $g(p)$   is  approximately determined by an explicit function 
	of $|p| $.

	\vspace{1mm}
	Let us now argue that the explicit function gives rise to a power series approximation. 
	To see this, recall that 
	the inverse hyperbolic tangent function is analytic in $( -1 , 1)$
	and its power series expansion is 
	$\tanh^{-1}(x) = \sum_{ 		 \substack{  n \in 2\mathbb N-1 	} }  	\frac{x^n}{n}$. 
	Thus, the generator  $g(p)$ is approximately given by a power series: Fixing $0\le \ve<1/2$, we have
	\begin{equation}
		\label{series g}
		g(p) =  \sum_{  j \in 2 \N_0	} 
		\alpha_j(\mu)  |p|^j   + O    \bigg(   \frac{   \Vu^2		}{\mu } \bigg)  \ ,  
		\quad 
		\alpha _j (\mu)
		: = 
		\frac{2^j \mu }{ (j+1) }
		\int_{\R^3}
		\frac{|V(k)|^2}{D(k)}
		\bigg( \frac{|k|}{D(k) }			\bigg)^j \d k \ 
	\end{equation}
with the error being \textit{uniform} over $|p|\le \ve \mu$. Note    that the coefficients   are positive $\alpha_j(\mu)>0$ and  satisfy     $\alpha_j (\mu) \sim \mu^{-j}$. 
	More precisely
	\begin{align}	\label{lem:g(p):expansion:general:case} 
			\lim_{\mu\rightarrow \infty}\mu^{j} \alpha_j (\mu) 
					=
					\frac{2^j}{(j+1) }
					\int_{\R^3} 
					\frac{  |V(k)|^2 }{|k|} \d k  \ . 
	\end{align}

	\vspace{1mm}

	By combining the results from Theorem \ref{thm 1} and 
	\eqref{series g} we can now verify Corollary \ref{cor:explicit:g} about the approximation
	of the original dynamics  with  the polynomial generators
	\begin{align}
		\h_g^{(N )}(p)
		  =
		p^2 - \sum_{ j \in 2 \N_0 :  j\le N  }
		\alpha_j(\mu)   |p|^j 	  \qquad N \in 2 \mathbb N _0  \ . 
	\end{align}

	\begin{proof}[Proof of Corollary \ref{cor:explicit:g}] 
		Let $ N \in 2 \mathbb  N_0$ and $\ve \in [0,\tfrac12)$ and note that due to Condition \ref{condition 1}, we can assume that $P_0 \le \ve \mu$. 
	First note that   the coefficients $\alpha_j(\mu)$ in  \eqref{series g}
		satisfy the upper bound 
			\begin{equation} 
			\alpha_j (\mu) 
			\leq 
			\frac{(2 / \mu)^j}{(j+1) }
			\int_{\R^3} 
			\frac{  |V(k)|^2 }{|k|} \d k   
			\label{bounds Kn}
		\end{equation}
		for all $ j \in 2 \N_0$ and $\mu>0$.
Thus, an  elementary estimate using the geometric series  shows that  \eqref{bounds Kn} implies  
		that for any $ J  \geq 2 $
		\begin{equation}
			\label{upper bound R}
			\sum_{j = J ;  j \t{ even }}^\infty \mathcal 
			\alpha_j(\mu) P_0^j 
		\, 	\leq \, 
		  \frac{  \| |k|^{-1/2}	V	\|^2 }{    ( 	1 - 2P_0/\mu)  	}
		\bigg(	\frac{2P_0}{\mu}					\bigg)^J  
 \, 	\equiv  \, 
		C_0 		\bigg(	\frac{2P_0}{\mu}					\bigg)^J   
		\end{equation}
Note that $C_0$ is independent of $J$. 
	Next, we use the  Duhamel formula, 
	the power series expansion \eqref{series g} with $\ve  = \frac{P_0}{\mu}$, 
	and the bound \eqref{upper bound R} with $J = N +2$
to find that    for some constant $C> 0 $
		\begin{align}
			\nonumber 
			\|	 e^{- i t \h_g}	\vp - e^{ - i t \h_g^{(N )}} 	 \vp \|  
			&  \, 	\leq  \, 
			|t| 	\, 	 	\|	 ( \h_g 	- \h_g^{(N ) } ) \vp  	\|	\\[2mm]
			& 	\, \le		\,   |t|		\,		 \bigg( \sum_{j=N +2}^\infty \alpha_j(\mu) \|\,  |p|^j \varphi \| + C \mu^{-1} \Vu^2 \bigg) \notag\\ 
			& 		\,  \leq 	\, 
			C  \, |t| \, 
			\bigg(
		\big(	2 P_0 \mu^{  -1  }			\big)^{N +2 }
			+  \mu^{-1} \Vu^2 
			\bigg)  \ . 
		\end{align}
		This finishes the proof  for $N <\infty$,   after we invoke Theorem \ref{thm 1} and 
		use the triangle inequality. Using Proposition \ref{prop expansion g}, the proof for $  N  = \infty $ follows directly from Theorem \ref{thm 1} and the triangle inequality.
	\end{proof}

	\subsection{Range of validity}
	\label{section:range}
	In the remainder of this section, we provide a detailed comparison of the range of validity of the  sequence of polynomial approximations, as stated in   Corollary \ref{cor:explicit:g}. 
	The following discussion can be viewed as an elaboration on Remark \ref{rem:non-conv}.

	\vspace{1mm}
	
	In order to extend the discussion further, 
	we parametrize 
	momentum and time scales 
	through two non-negative numbers $(a,b)$ as follows 
	\begin{align}
		\varphi  \in  \mathds 1(\tfrac12 \mu^a \le |p| \le  \mu^a) L^2_X   \quad \text{and} \quad t \sim \mu^b, 
	\end{align}
	where $\vp\in L^2_X$ is the initial datum. 
	In this context,   our strongest result in the massless case is Theorem \ref{thm 1}, 
	which proves  $ \lim_{\mu \to \infty} ( e^{-i t \h_g}\vp\otimes \Omega - \Psi(t) ) = 0$  
	for all $(a,b)\in I := [0,1) \times [0,1/2)$.

	\vspace{1mm}

	Certainly,  employing  the polynomial generators $\h_g^{(N )}(p)$ as an effective dynamics  has 
	the advantage of being  explicit in comparison to $\h_g(p)$.
	However, the  truncation of the series introduces an additional error term in Corollary \ref{cor:explicit:g}, relative to Theorem \ref{thm 1}.
	Thus,  convergence  is only guaranteed 
	for parameters $(a,b)\in I_N $ belonging to the smaller regions
	\begin{align}
		I_N := \bigcup_{a\in [0,1)} \big[0,a\big] \times \big[0 ,\min(1/2, (N+2)(1-a)) \big).
	\end{align}
	Note that as $N\rightarrow \infty$, the sets $I_N$ start to cover the full window $I$. For illustration  purposes, some of these 
	regions are sketched in Figure  2. 
	\begin{figure}[t]
		\label{fig1}
		\includegraphics[width=15cm]{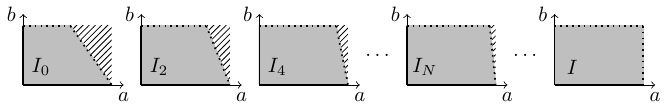}
		\centering
		\caption{  	The $(a,b)$-regions $I_N$}
	\end{figure}

	\vspace{1mm}
	
	The next theorem demonstrates that these smaller $(a,b)$-regions are in fact optimal. More concretely, we show that the approximation using $\h_g^{(N)}(p)$ fails on the diagonal boundary of $I_N$.
	That is, for $b(a) =  (N+2)(1-a) $ with  $a\in (1 -  \frac{1}{2(N+2)} , 1 ) $. 
	
	\vspace{1mm }
Here for simplicity we consider the massless Nelson model 	  \eqref{eq:Nelson:form:factor}, so that $\Vu^2 = O(\log \mu)$.

	\begin{theorem}\label{thm:non-convergence} 
Let  $V$ be given by \eqref{eq:Nelson:form:factor} and choose 
 $ N \in 2 \mathbb N_0$. Further, assume that 
	\begin{equation}
\label{initialdatum}
		 \varphi 
		 =
		  \mathds 1(\tfrac12 \mu^a \le |p| \le \mu^a) \varphi  \qquad \t{ with  } \qquad 
		 a  \in \bigg( 1  -  \frac{1}{2(N +1)}  \, ,	\, 1 		\bigg)	  \ . 
	\end{equation}
		Then, there exists a constant
		$\tau = \tau (N )>0$ such that for $t = \tau \mu^b$ with $b= ( N  + 2 ) (1-a) $
		\begin{align}
			\liminf_{\mu \to \infty} \Big\|\Big( e^{-i t H} - e^{-i t \h_g^{(N )}(p) } \Big) \varphi \otimes \Omega \Big\| > 0.
		\end{align}
	\end{theorem}
	\begin{remark}\label{rem:non-conv-2} To continue the example from Remark \ref{rem:non-conv}, 
	 consider  an initial state  $\vp$
	 satisfying  \eqref{initialdatum} with $a = 23/24$.  
		Remember that Corollary \ref{cor:explicit:g} proves 
		\begin{align}\label{rem:cor}
			\lim_{\mu \to \infty} (e^{ -itH} - e^{-it\h_g^{(N)}(t)})  \varphi \otimes \Omega = 0 \quad \text{for all} \quad |t| \ll \min\big\{ \mu^{1/2} \log(\mu)^{-1} , \mu^{(N+2)/24} \big\}.
		\end{align}
		By Theorem \ref{thm:non-convergence}, on the other hand, we know that  $(e^{-itH} -  e^{-it \h_g^{(N)}(p)} ) \varphi \otimes \Omega$ does not converge to zero for $t\sim \mu^{(N+2)/24}$. Thus, we can conclude that \eqref{rem:cor} is in fact optimal. Note that in the considered case with $a=23/24$, we hit the threshold at $N =10$, for which the time scale does not extend further by considering larger values of $N$. 
		The discussion is easily generalized to other choices of $P_0= \mu^a$.
	\end{remark}
	
	\begin{proof}[Proof of Theorem \ref{thm:non-convergence}]
		Let $  N  \in  2 \mathbb N_0$. 
		Throughout this proof we consider the remainder term 
		\begin{equation}
			\mathcal R_N (p) 
			= \sum_{ j \ge N+2 : j \t{ even }		}  
			\alpha_j(\mu) |p|^j \ 		\quad  	\t{on} \quad \tfrac12 \mu^a \le |p| \le \mu^a
		\end{equation}
		corresponding to the truncation of the series in \eqref{series g} of order $n$. 
		In view of \eqref{upper bound R} 
		we have the following upper bound for
		$|p|\leq   \mu^a$ with $a <1$ and $\mu>0$ large enough
		\begin{equation}
			\label{eq:bound:on:R}
			\mathcal R_ N (p)   \leq 
			C_1	 \mu^{(N +2) (a-1)} \  , 
		\end{equation}
		for a constant $C_1 = C_1 ( N )$. 
		On the other hand, under the assumption $|p|\geq  \frac{1}{2} \mu^a$, 
		it follows from  \eqref{lem:g(p):expansion:general:case} 
		that the following lower bounds holds
		for $\mu>0$ large enough 
		\begin{equation}
			\label{eq:ref:lb:R}
			\mathcal R_N (p)  
			\geq \alpha_{N +2}(\mu)	|p|^{N +2}
			\geq 
			C_2 
			\mu^{  ( N +2) (a-1)} 
		\end{equation}
		for a constant $C_2   \geq C( N )  \int |k|^{-1} |V(k)|^2 \d k>0$.

		Next, we use these bounds to derive a lower bound for  $(e^{-i t \h_g(p) } - e^{-i \h_g^{(N  )}(p)} ) \varphi$. To this end, we invoke Duhamel's formula twice, that is
		\begin{align}
			i\Big( 1- e^{ i t \h_g(p) }  e^{-i t \h_g^{( N )}(p) } \Big) \varphi & = \int_0^t \d t e^{is \h_g(p) } e^{-is\h_g^{(N )}(p)} \Big( \h_g(p) - \h_g^{(N  )}(p) \Big) \varphi \notag\\
			& = t \Big( \h_g(p) - \h_g^{(N )}(p) \Big) \varphi \notag\\
			& \quad - \int_0^t \d s \int_0^s \d r \,  e^{i(r-s) \h_g (p) } e^{-ir\h_g^{(N )}(p)} \Big( \h_g(p) - \h_g^{(N )}(p) \Big)^2 \varphi.
		\end{align}
		In view of Proposition \ref{prop expansion g} and    \ref{series g}, we have (the ratio  $\ve = \mu^a/\mu\ll 1$ can be disregarded)
		\begin{align}
			\h_g(p) - \h^{(N )}_g(p) =   \mathcal R_ N (p) +  O(\mu^{-1} \log(\mu) )
		\end{align}
		and thus, using \eqref{eq:bound:on:R},  \eqref{eq:ref:lb:R}  and $\mu^{-1}\log(\mu) = o (	 \mu^{(N +2)(1-a)}	)$ as $\mu\to\infty$,	we can estimate
		\begin{align}
			\Big\| \Big( \h_g(p) - \h_g^{(N )}(p) \Big)^2 \varphi \Big\| & \le  D_V  \big( \mu^{(N +2)(a-1) }\big)^2  \\
			\Big\| \Big( \h_g(p) - \h_g^{(N )}(p) \Big) \varphi \Big\| & \ge C_V \mu^{(N +2)(a-1)}
		\end{align}
		for suitable constants $D_V,C_V>0$.
		Choosing $t = \tau \mu^{b}$ with $b=( N +2)(1-a)$ and $\tau^2 = C_V / D_V >0 $, we arrive at
		\begin{align}
			\Big\| 
			\Big( 1- e^{ i t \h_g(p) }  e^{-i t \h_g^{(N )}(p) } \Big) \varphi \Big\| & \ge |t| \Big\| \Big( \h_g(p) - \h_g^{(N )}(p) \Big) \varphi \Big\| - \frac{t^2}{2} \Big\|  \Big( \h_g(p) - \h_g^{(N )}(p)\Big)^2 \varphi \Big\| \notag\\
			& \ge   C_V   \tau - \frac{ D_V}{2} \tau^2 =  \frac{C_V}{2}> 0 \ . 
		\end{align}
		Together with $\limsup_{\mu\to \infty} \| ( e^{-itH } - e^{-i t  \h_g(p) } ) \varphi \otimes \Omega \| = 0 $  by Theorem \ref{thm 1}, this proves the claimed statement.
	\end{proof}

	\section{Effective dynamics for massive fields} 
	\label{section massive}
	The  main   goal  of this section is to  prove  Theorem \ref{thm 2}. 
	To this end,   we consider a massive dispersion relation 
	\begin{equation}
		\omega(k) =   \sqrt{   k^2  + m^2 }  \ 
	\end{equation}
with mass term $m \ge \mu^{- 1 }$. 
In what follows, 
	$\Psi_0  $  denotes the 
	initial datum satisfying Condition \ref{condition 1}
	relative to $0 < \ve < 1/2$, 
	and $V$ is the form factor satisfying Condition \ref{condition 2}. 
	We let  $\Psi_{\eff}(t) = e^{ - i t \h_g}\vp \otimes \Omega$  be 
	the effective dynamics, 
	with Hamiltonian  $\h_g = p^2  -  g(p)$. 
	The generator $g : \RE^3 \rightarrow \RE $
	is assumed to be bounded and measurable, and shall be chosen 
	below. 
	
	\vspace{1mm}
	
	The first step is to  recall  from Section \ref{section prelim}		that the following identity holds  for all $t \in \RE$
	\begin{equation}
		\label{eq 1 section 6}
		e^{i tH } \Big(  		  \Psi(t) - 	\Psi_{\eff}(t)    \Big) 
		=   
		\I(t) 
		\Big(
		\V + g(p) 
		\Big)
		\Psi_0  \ . 
	\end{equation}
	Recalling our notation $\P_\Omega   $ 
	and $	\Q_\Omega    $ from \eqref{def:projections}, 
	an application of the  integration by parts identity \eqref{I(t)} yields 
	\begin{align}
		\label{eq 2 section 6}
		\mathcal I(t) \Q_\Omega \V & = \B(t) \R \V + \mathcal I(t)  \V \R \V   \notag\\
		& = \B(t) \R \V +   \mathcal I(t) \P_\Omega \V \R \V   + \mathcal I(t) \Q_\Omega \V \R \V     \ .
	\end{align}
	Therefore, we combine \eqref{eq 1 section 6} and \eqref{eq 2 section 6} to find  that 
	\begin{equation}
		\label{eq 3 section 6}
		e^{i tH } \Big( 		  \Psi(t) - 	\Psi_{\eff}(t)     \Big) 
		= 
		\B(t) \R\V \Psi_0 
		+  
		\I (t) \Big(
		\P_\Omega \V \R \V \Psi_0   +  g(p) \Psi_0 
		\Big)
		+	 
		\I(t) \Q_\Omega \V \R \V \Psi_0  \ . 
	\end{equation}
	The next step is note that the expansion process can be continued, 
	if  we  now apply the identity  \eqref{eq 2 section 6} 
	to the last term in \eqref{eq 3 section 6}
	and expand $\I(t) \Q_\Omega \R \V \R $ again.
	This process can be repeated arbitrarily many times. 
	In particular, for $ n \in \mathbb N $ we find 
	\begin{align} 
		\label{eq 4 section 6}
		e^{i tH } \Big(    &  		  \Psi(t) - 	\Psi_{\eff}(t)    \Big)  \\ 
		& = 
		\sum_{1 \leq j \leq n } \B(t) (\R \V)^j \Psi_0  
		+
		\I(t) 
		\bigg(
		\sum_{		\substack{1 \leq j \leq n		\\		 j \t{ odd}}   }
		\P_\Omega
		\V ( \R \V )^j  \Psi_0  +  g(p)\Psi_0
		\bigg)	 
		+  \mathcal I(t) \Q_\Omega  \V (\R \V)^{n } \Psi_0  \   ,  \nonumber 
	\end{align}
	where 
	we have used the 
	fact that 
	for  any even  $j \in  \mathbb N :$   $\P_\Omega \V (\R\V)^{j } \P_\Omega  =0$.

	In the remainder of the section, we show that one can choose a generator 
	$g(p)= g_n(p)$
	such that the middle term in \eqref{eq 4 section 6} vanishes.
	Every other term can be regarded as an error term, with precise estimates derived in Proposition \ref{prop rem}. In Section \ref{sec: proof thm2}, the results are combined to prove Theorem \ref{thm 2}.
	
	\subsection{Choosing   the generator}
	Our present goal 
	is to prove that by choosing $g  = g_n$
	as in Definition \ref{definition gn}, the middle term  of \eqref{eq 4 section 6} cancels exactly. 
	This is the content  of  Lemma  \ref{lemma 1 section 6} stated below. 

\vspace{1mm}

To this end,  let us  recall    the following notation
	for a special class of collection of sequences $\sigma$ of length $  j   \in \mathbb N $. 
	\begin{align}
		\Sigma_0  ( j  ) 
		& 	 = 
		\textstyle 
		\{	   \sigma \in \{ 	 +1  , -1 \}^j  \ :  \  
			\sum_{ i  = 1}^{\ell }	\sigma (i) \geq 1 \ \forall \ell \leq j    -1  \ \  \t{and} 	  \ \   \sum_{ i=1 }^j \sigma (i) = 0  			\}   \ . 
	\end{align}
	Next, we introduce auxiliary functions
	that will help us navigate the proof of the upcoming Lemma. 
			More precisely, 	let $n \in \mathbb N $ 
	and  $ | p| <   \frac{1}{2 } \mu$.
	We 
	define     $F_n		( p , \cdot )	: (0, \infty) \rightarrow   (0 , \infty )	$ as
	\begin{align}
\label{Fn}
		F_n ( p , x)  &
		: =
		\sum_{	\substack{	 j =2 	 \\  j \t{ even }	 }	}^{n+1 } 
		\sum_{\sigma \in \Sigma_0(j)}	 F_{ j }^\sigma ( p , x) 	\\ 
		\nonumber 
		F^{\sigma}_{j } ( p , x) 
		&:  = 
		\mu^{j/2}
		\int_{\RE^{3  j }} 
		\d k_1 \cdots \d k_{j}
		V(k_1) \cdots V(k_{j} )
		\, \langle \Omega , b_{k_{j}}^{\#_{\sigma(j )}} \ldots  b_{k_1}^{\#_{\sigma(1)}} \Omega\rangle  \\
		& \qquad  \times   
		\prod_{ \ell =1}^{ j  -1  } 
		\bigg( 
	 \mu 	\sum_{ i =1}^\ell \sigma(i) \omega(k_i) 
		+ 
		\Big( 	 p - \sum_{i=1}^\ell \sigma(i) k_i		\Big)^2
		- 
		p^2  +x 
		\bigg)^{-1} \ . 
	\end{align}
In particular, we may   re-phrase the definition of $g_n(p)$ given in Definition \ref{definition gn} 
in terms as solutions of the fixed point equation 
		\begin{equation}
		g_n(p) =  	F_n ( p  ,   g_n(p))  
	\end{equation}
	for all $| p | <   \frac{1}{2} \mu$.

	\begin{lemma}\label{lemma 1 section 6}
		Let $n \in \mathbb N $ and let $g = g_n$
		be as in Definition \ref{definition gn}. 
		Then, it holds that  
		\begin{equation}
			\sum_{\substack{1 \leq j \leq n		\\		 j \t{ odd}}   }   
			\P_\Omega    \V ( \R \V )^j  \Psi_0     = -   ( g(p) \otimes \1 ) \Psi_0 \ . 
		\end{equation}
	\end{lemma} 
	
	\begin{proof}
		Throughout the proof, we treat $g(p)$ as an arbitrary real-valued, bounded measurable function on $\RE^3$. The precise choice given by Definition \ref{definition gn} will only be made at the   end. 
 For transparency, we denote the $g$-dependent resolvent 
		on $ \Q_\Omega\H $  
		\begin{equation}
			\mathbb R_g  := 	  ( 
			\h_g(p+P)
			-p^2-T 
			  )^{-1}  
		\end{equation}
		We will consider on 
		the Hilbert space $\mathscr{H} = L_X^2 \otimes \F$  
		the following operator, regarded as a function of $g$: 
		\begin{align} 
			\label{eq 5 section 6}
			\mathcal A _n [g]
			:=
			\sum_{	\substack{	 j =2 	 \\  j \t{ even }	 }	}^{n+1 } 
			A_{j}[g]
			\quad \text{with} \quad 
			A_{j} [g]  :=   \P_\Omega \V (\R_g  
  \Q _\Omega 
			 \V)^{j -1 }  \P_\Omega 
		\end{align}
		which coincides with the operator in \eqref{eq 4 section 6}
		after   the change of variables   $j\mapsto j-1$.

		\vspace{1mm}
		
		First, we use the decomposition $\V = \V^+ + \V^+$
		into creation and annihilation operators, 
		in each factor  of $A_j[g]$.
		The expansion can be represented as a sum 
		over sequences $\sigma \in \Sigma_0(j)$ as follows 
		\begin{align}
			\label{eq 7 section 6}
			A_{j}[g]
			=
			\sum_{\sigma \in \Sigma_0(j)} A_{j}^\sigma [g] ,
			\quad 
			A_j^\sigma  [g]
			:= 
			\P_\Omega		
			\V^{\sigma(j)}  \R_g \V^{\sigma(j-1)}   \cdots  \R_g  \V^{\sigma(1)} 
			\P_\Omega \ . 
		\end{align} 
	Here, we have dropped the projections 
$\Q_\Omega $ 
in \eqref{eq 5 section 6}
at the expense of summing over the sequences $\sigma \in \Sigma_0( j  )$.
The
 summability condition 
 $	\sum_{ i  = 1}^{ \ell  }	\sigma (i) \geq 1 \ \forall \ell   \leq j    -1  $ 
 guarantees that $\R_g$ acts on states that contain at least one particle. 

\vspace{1mm}

		The next step is to   compute the operator  
		$	 A^\sigma_j [g]$
		using commutation relations between $\R_g$
		and $e^{ - i kX }$ and $b_k$ and $b_k^*$: 
		\begin{align}
			\R_g 
			e^{   i 	  k_i X }
			b_{k_i} 
			& = 
			e^{  i 	   k_i X }
			b_{k_i} 
			\big(
			\h_g(p+P) - 
			(p  + k_i )^2
			- T + \mu \omega(k_i) 
			\big)   \\
			\R_g 
			e^{    - i 	  k_i X }
			b_{ k_i}^*
			& = 
			e^{  -  i 	   k_i X }
			b_{k_i}^* 
			\big(
				\h_g(p+P) - 
			(p  -  k_i )^2
			- T -  \mu \omega(k_i)  
			\big)  \ . 
		\end{align}
		The idea is to   move all  the factors $e^{ikX}b_k$ and $e^{- i kX}b_k^*$ to the \textit{left} of all the resolvents, in the operator 
		\eqref{eq 7 section 6}. 
		Once they hit the vacuum projection $\P_\Omega$
		on the right, the shifted resolvents
		are evaluated using 
		$T \Omega = P \Omega  =  0 $.
		The result is an operator that only acts
		on the $L_X^2$ variables, 
		tensored with the monomial 
		$b_{k_j}^{\#_{\sigma( j )}} \ldots  
		b_{k_1}^{\#_{\sigma(1)}}$. 
		One then takes the vacuum expectation value over such monomial. 
		This calculation yields  
		(here $p = - i \nabla_X$ and $\int = \int_{\RE^{3j}} \d k_1 \cdots \d k_ j $)
		\begin{align}
			A_j^\sigma  [g]
			= 
			\mu^{j/2}
			\int 
			V(k_1) \cdots V(k_j )
			\prod_{ \ell =1}^{ j -1}
			R_g^{\sigma(1) \cdots  \sigma (\ell )}
			( p ; k_1, \cdots, k_\ell)
			\langle \Omega , 
			b_{k_j}^{\#_{\sigma( j )}} \ldots  
			b_{k_1}^{\#_{\sigma(1)}} \Omega\rangle 
			\otimes \1 
		\end{align}
		where we introduce
		the  following operators
		on $L_X^2 $
		\begin{align} 
			R_{g}^{\sigma(1)\ldots\sigma(\ell)}( p ; k_1,\ldots,k_\ell) 
			& 	= 
			(-1) 
			\textstyle
			\bigg(
				 \mu 
			\sum_{ i =1}^\ell \sigma(i)    \omega(k_i)  
			+
							\Big( 	 p - \sum_{i=1}^\ell \sigma(i) k_i		\Big)^2
							-p^2 + g(p)  	\bigg)^{-1} 
		\end{align}

		Finally, in terms of the auxiliary functions
		introduced in  \eqref{Fn}
		we identify  that 
		$A_j^\sigma [g]  \Psi_0 =   (-1)^{j-1 } F_j^\sigma(p, g(p)) \otimes \1 \Psi_0 $
		where the right hand side is evaluated 
		on the operator $p = - i \nabla_X$ in the
		spectral subspace $\1 (|p |\leq \ve \mu)$. 
We sum over all sequences $\sigma \in \Sigma_0(j)$
		and all  even $ 2 \leq j \leq n+1 $
		to obtain
		\begin{equation}
			\mathcal A_n[g] 
			\Psi_0  =  
			\sum_{	\substack{	 j =2 	 \\  j \t{ even }	 }	}^{n+1 } 
			A_{j}^\sigma [g] \Psi_0 
			= 
			 - 
			\sum_{	\substack{	 j =2 	 \\  j \t{ even }	 }	}^{n+1 } 
			F_{j}^\sigma (p ,  g(p)) \otimes \1 \,  \Psi_0 
			=  - F_n(	 p ,   g(p)) \otimes \1 \, \Psi_0  \ . 
		\end{equation}
		Therefore, by choosing $g = g_n$ as in Definition \ref{definition gn}
		we see that 
		$ 	\mathcal A_n[g_n]  \Psi_0 =  - g_n(p) \otimes \1 \Psi_0$.
	This finishes the proof. 
	\end{proof}

	\subsection{Estimate for $(\R \V)^n \Psi_0$}
	\label{subsection6.1}
	In this subsection, 
	we state some estimates that we will need in order to control the first and third term in the expansion 
	\eqref{eq 4 section 6}. 
	The main challenge is to estimate  
	terms of the form 
	\begin{equation}
		\label{prod}
		\underbrace{	\R \V^{\#} \cdots \R \V^{\# } 	}_{n \t{ times }}	 (\vp \otimes \Omega)   \ . 
	\end{equation}
	Here,  $\V^\#$ can be either  a creation $\V^+$ or annihilation  $\V^-$ operator   (see \eqref{decomposition} for a definition)
	and     $ \R    $ is the resolvent operator, introduced in Definition \ref{def R and B}.
	It is crucial to observe that the presence of the projections $\Q_\Omega$ in the definition of $\R$
	guarantees that the states \eqref{prod}
	are either zero, or contain at least  one  particle. 
	
	\vspace{2mm}
	
	In this regard,  our most relevant     result is 
	Proposition \ref{prop rem}  stated below. 
	In order to formulate it, 
	we introduce some some notation that encodes the structure of the states \eqref{prod}.
	Let us  consider the  
	set of sequences of length $n$, taking values on $\{ + 1, -1\}$, 
	and whose all partial sums are bounded below by one. 
	That is, 
	\begin{equation}
		\Sigma(n) : = 
		\textstyle 
		\{  \sigma \in \{+1  , -1 \}^n	 \  : \  	
		\sum_{i=1}^j  \sigma (i)  \geq 1 	 \ \forall  j = 1 , \ldots, n	\} \  , \qquad n \in  \mathbb N \ . 
	\end{equation}
Contrary to $\Sigma_0(n)$, the set of sequences in $\Sigma(n)$ do not drop to zero at the last entry.

	\vspace{1mm}
	The proof of the following proposition is rather involved 
	and   will    be postponed to Section~\ref{section resolvent estim}.   
	We remind the reader that we assume  $m \geq   \mu^{-1 }$. 
Here, we employ the notation $\V^{+1} \equiv \V^{+}$
	and $\V^{-1 } \equiv \V^-$, and
	$n_\pm  \equiv  | \{i :  \sigma(i) = \pm 1 \}|$
		whenever $\sigma$ is known from context.

	\begin{proposition}\label{prop rem}
		Let $  n \geq 2$. 
		Then, there exists a constant $C>0$ such that for all sequences   $\sigma  \in \Sigma (n)$ 
		\begin{equation}
			\label{prop rem estimate}
			\bigg\|	
			\prod_{i=1}^n	(\R  \V^{\sigma (i) })	 (\vp \otimes \Omega)				\bigg\|_{	 \mathscr H }
			\leq 
			\frac{C^n}{\mu^{n/2}}
			\sqrt{( n_+ - n_-  )   !} 
			\bigg( \frac{1}{m}	\bigg)^{n_-}
			\Vu^{n_+}
			\|	  V	\|_{L^2}^{n_- }  \    
		\end{equation} 
		for all $\mu>0$ large enough. 
	\end{proposition}

	\begin{remark}
		Two comments are in order 
		
		\begin{enumerate}[leftmargin=*]
			\item 
			The summability conditions implies $n_+ - n_- \geq 1 $. 
			On the other hand, $ n = n_- + n_+$. 
			Consequently, $2 n_- \leq  n -1  $  
			or equivalently
			$n_- \leq 		\lfloor	\frac{n-1}{2}  	 \rfloor $ 
			represents the  
			worst possible scenario, 
			regarding the growth with respect to $m$. 
			In particular, we obtain  the following bound that is  \textit{uniform} over  all  $\sigma  \in \Sigma(n) $  
			\begin{equation}
			\bigg\|	
				\prod_{i=1}^n	(\R  \V^{\sigma (i) })	 (\vp \otimes \Omega)				\bigg\|_{	 \mathscr H }
				\leq  
				\frac{C^n \sqrt{n! }		}{\mu^{n/2}}
				\bigg( \frac{1}{m}	\bigg)^{ 	\lfloor	\frac{n-1}{2}  	 \rfloor } 
				\Vu^n  
				\ . 
				\label{estimate n state}
			\end{equation} 
 where we used\footnote{
Indeed,  the 
Monotone Convergence Theorem implies  
 $ \| 	V\|_\mu \geq 2  \|	|k|^{-1/2}	 V\|_{L^2 }$ 
 for  $\mu$ large enough. 
 Additionally,   note that   $ C_0 :=  \|	|k|^{-1/2}	 V\|_{L^2 } / \|	V\|_{L^2}>0$
is  itself   by definition   a constant. Thus, take $C = 2C_0$.
} $ \|	V\|_{L^2} \leq  C \|	V\|_\mu$.  This bound, although not optimal with respect to $n\in \N$, 
			will be useful in the proof of the main theorem.

			\vspace{1mm }

			\item 
			Consider a $\mu$-dependent mass term $m = \mu^{- \delta  }$ where $ \delta \in (0,1) $
			and let $n$ be odd. 
			Then,  the   right hand side of \eqref{estimate n state} 
			is of order $\mu^{	- \frac{1}{2}	(	n  (1- \delta ) 	  -  \delta )	} $. 
			In particular, the rate gets better with the    order of iteration $ n \geq 1 $, but this growth   is modulated by $\delta $, i.e. the size of the mass term. 
			Further, the case $ m = \mu^{-1} $ is \textit{critical}:  
			our bound on the  rate of convergence does not improve with $n\geq 1 $ but rather stagnates.

		\end{enumerate}
		
	\end{remark}

	\subsection{Proof of Theorem \ref{thm 2}}\label{sec: proof thm2}
We now have all the necessary ingredients to prove
our main approximation result for massive fields.

	\begin{proof}[Proof of Theorem \ref{thm 2}]
		Let $\Psi_{	\eff}^{(n)} (t) = e^{- i t \h_{g_n}		}\vp \otimes \Omega$
		where the generator  $g_n$
		is chosen as in Definition \ref{definition gn}. 
		It  follows immediately from \eqref{eq 4 section 6},  Lemma \ref{lemma 1 section 6}
		and  $\|	\B(t) \|\leq 2 $    that 
		\begin{align} 
			\label{eq 8 section 6}
			\| 	\Psi^{(n)}_{	\eff}(t)    -  \Psi(t ) \|     \leq 
			2	\sum_{1 \leq j \leq n }   \|  	   (\R \V)^j \Psi_0  \| 
			+   \|	 \I (t) \Q_\Omega  \V ( \R \V )^n \Psi_0 	\| \ . 
		\end{align}

		Next, let us    note that since  $  n $ is odd,   the state $\Psi_n : = \Q \V (\R \V)^n\Psi_0 $
		has a trivial one-particle sector, and a non-trivial two-particle sector. In particular, $\P_\Omega \V \R \Psi_n = 0$.
		Thus, integrating by parts one more time (i.e. apply \eqref{eq 2 section 6}) leads to
		\begin{align}
			\label{I1}
		\I(t) \Psi_n 
			& =
			\B(t) \R \Psi_n  
			+ 
			\I(t)	 \V \R \Psi_n   =
			\B(t) \R \Psi_n  
			+ 
			\I(t)	\Q_\Omega \V \R \Psi_n  
\end{align}
Below we will use this identity in the case $\mu^{-1} \le m  \leq \mu^{-1 /2 } \Vu $, which we call the small mass case. Conversely, in the opposite large mass case $m \geq \mu^{-1 /2 }  \Vu  $, it is favourable to apply the integration by parts formula one more time, that is
\begin{align}	\label{I2}
			\I(t) \Psi_n 
			& 	  =
			\B(t) \R \Psi_n + \B(t) \R\V \R \Psi_n +
			\I(t) \V \R \V \R \Psi_n  \ . 
		\end{align}
For the remainder estimates, note that in the first case,  the relevant state 
		is   
		$\R \Psi_n = ( \R\V )^{n+1} \Psi_0 $
		which contains at most $n+1$ particles, an even number. 
		In the second case, the 
		relevant state is 
		$ \R \V \R \Psi_n = ( \R\V )^{n+2} \Psi_0 $
		which 
		contains at most $n+2$ particles, an odd number. 
		This on its turn implies that the least integer bracket in the estimate   \eqref{estimate n state}
		takes different values
		and thus makes the dependence on the mass term different. 
		
		\vspace{1mm}
		
		\textit{The small mass case}. 
		An    application of 	
		\eqref{eq 8 section 6}, 
		  \eqref{I1}  and 
		\eqref{estimate n state} 
		yields
		\begin{align}
			\| 	\Psi^{(n)}_{	\eff}(t)    -  \Psi(t ) \|     
			& 	\leq 
			2 	\sum_{ j =1 }^{n+1 }   \|  	   (\R \V)^j \Psi_0  \|  
			+ 
			\, 		 | 	t | \,  \mu^{1/2} \, 
			 \|	V\|_{L^2}  (n+2 )^{\frac{1}{2}}
			\|	 ( \R \V 	)^{n+1} \Psi_0	\| 	 	\\ 
			& \leq 
			\sum_{ j =1 }^{n+1 }  
			\frac{C^j \sqrt{j!}	}{\mu^{j/2 }} 
			\bigg(		 \frac{1}{m}			\bigg)^{	 \lfloor \frac{ j -1 }{2 } \rfloor 		}
		\Vu^j 
			+ 
\, 		 | 	t | \, 
\sqrt{ (n \! +\! 2)!}	  \, 
			\frac{C^{n+1}  }{\mu^{ \frac{n}{2}}} 
			\bigg(		 \frac{1}{m}			\bigg)^{	 \frac{n- 1 }{2}	}
		\Vu^{n+1}
			\nonumber 
		\end{align}
		where in the last line we have evaluated 
		$ \lfloor \frac{  n  }{2 } \rfloor  = \frac{n-1}{2}	$
		for $n$ odd. 
		The above bound records the dependence of the estimates with respect to $ n \in \mathbb N$. 
		Let us now simplify it. 
		First, without loss of generality we assume    $\mu^{-1/2}\Vu \le 1$  (otherwise
		   the bound  in 
 Theorem \ref{thm 2}  holds trivially). 
Next, for  the $t$-independent error  
 we estimate the summands   using $ C^j \sqrt{j!} \leq C^{n+1} \sqrt{(n+1)! } $ as well as the following inequalities (recall $n$ is odd):
\begin{align}
 \sum_{ j=1}^{n+1}
 			\frac{	\Vu^j }{\mu^{j/2 }} 
			\bigg(		 \frac{1}{m}			\bigg)^{	 \lfloor \frac{ j -1 }{2 } \rfloor 		}
	  \notag  
	&	 = 
	  \sum_{\substack{ j =1  \\ j\, \text{odd} }}^{n}
\frac{	\Vu^j }{\mu^{j/2 }}  
			\bigg(		 \frac{1}{m}			\bigg)^{	  \frac{ j -1 }{2 }  		} 
				+
		  		  \sum_{\substack{ j =2  \\ j\, \text{even } }}^{n+1 }
		  		\frac{	\Vu^j }{\mu^{j/2 }}  
		  		\bigg(		 \frac{1}{m}			\bigg)^{	  \frac{ j -2 }{2 }  		}   \notag\\
			&    = 
			  \sum_{\substack{ j =1  \\ j\, \text{odd} }}^{n}
			\frac{	\Vu^j }{\mu^{j/2 }}  
			\bigg(		 \frac{1}{m}			\bigg)^{	  \frac{ j -1 }{2 }  		} 
			+
\frac{ \Vu 	}{\mu^{1/2}}
			\sum_{\substack{ j =1  \\ j\, \text{odd } }}^{n  }
			\frac{	\Vu^j }{\mu^{j/2 }}  
			\bigg(		 \frac{1}{m}			\bigg)^{	  \frac{ j -1 }{2 }  		}   \notag\\
		 &   =  
\bigg( 1 + 		 \frac{\Vu }{\mu^{1/2}} \bigg) 
\sqrt{m}
		 	  \sum_{\substack{ j =1  \\ j\, \text{odd} }}^{n}
		 \frac{	\Vu^j }{\mu^{j/2 }}  
		 \nonumber 
		 \bigg(		 \frac{1}{m}			\bigg)^{	  \frac{ j   }{2 }  		}   \\ 
& 		 \leq 	
		 2 n  \bigg(	 
		 	\frac{\Vu 	}{\mu^{1/2}}			+ \frac{\Vu^n}{\mu^{n/2}}   
		 	\bigg(		 \frac{1}{m}			\bigg)^{	  \frac{ n -1 }{2 }  		} 
		 \bigg)	 \ . 
		 \label{estimate:lowmass}
\end{align}
In the first line,   we evaluated even and odd terms; in the second 
line  we shifted   the second sum $j\to j+1$ and in third  line we grouped terms. 
In the last line  we used 
 $\mu^{-1/2} \Vu \le 1$ and elementary geometric series estimates. 
This gives    (here we absorb $n  \leq C^n$)
\begin{align} 
\label{small mass case}
			\| 	\Psi^{(n)}_{	\eff}(t)    -  \Psi(t ) \|     
			&  \ 	\leq   \ 
			C^n	\,   \sqrt{n! } \, 
			\bigg(  \   
			\bigg[  
			\frac{\Vu 	}{\mu^{1/2}}		 \, 	+	\, 	 \frac{\Vu^n}{\mu^{n/2}}   
		\bigg(		 \frac{1}{m}			\bigg)^{	  \frac{ n -1 }{2 }  		} 
		\bigg] 
		   \ + \ 		 \frac{|t| 	}{	 \mu^{1/2}	}
						\frac{ \Vu^{n+1} }{    (	  \mu m	)^{  \frac{n-1}{2}}			 }
		 \ 	\bigg)
\end{align}
		
		\textit{The large mass case}. 
		An  application of \eqref{eq 8 section 6}, 
		\eqref{I2}
		and \eqref{estimate n state}
		yields
		\begin{align}
\label{large mass estimate}
			\| 	\Psi^{(n)}_{	\eff}(t)    -  \Psi(t ) \|     
			& 	\leq 
			2	\sum_{1 \leq j \leq n +2  }   \|  	   (\R \V)^j \Psi_0  \|  
			+ 
\, 		 | 	t | \,   \mu^{\frac{1}{2}}		\, 
			 \|	V\|_{L^2}  (n+3 )^{\frac{1}{2}}
			\|	 ( \R \V 	)^{n+2} \Psi_0	\| 	 	\\ 
			& \leq 
			\sum_{1 \leq j \leq n+2}
			\frac{C^j \sqrt{j!}	}{\mu^{j/2}} 
			\bigg(		 \frac{1}{m}			\bigg)^{	 \lfloor \frac{ j -1 }{2 } \rfloor 		}
		\Vu^j 
			+ 
\, 		 | 	t |  	\,
			\frac{C^{n+2} \sqrt{ (n \! +\! 3)!}	}{\mu^{ (n+1)/2  }} 
			\bigg(		 \frac{1}{m}			\bigg)^{	 \frac{n +  1 }{2}	}
		\Vu^{n+2}
			\nonumber 
		\end{align}
		where in the last line we have evaluated 
		$ \lfloor \frac{  n +1  }{2 } \rfloor  = \frac{n+1 }{2}	$
		for $n$ odd. 
		In order to simplify the  $t$-independent term  we use $C^j \sqrt{j!} \leq C^n \sqrt{n! }$
		for the summands. Additionally, we   use   the calculation \eqref{estimate:lowmass}
	and  include  an additional  term to the sum to find 
	\begin{equation}
		 \sum_{ j=1}^{n+2}
		\frac{	\Vu^j }{\mu^{j/2 }} 
		\bigg(		 \frac{1}{m}			\bigg)^{	 \lfloor \frac{ j -1 }{2 } \rfloor 		} 
		\leq 
			 2 ( n +1 ) \bigg(	 	\frac{\Vu 	}{\mu^{1/2}}			+ 
			 \frac{\Vu^{n+2 }}{\mu^{	 \frac{n+2}{2}	}}   
		\bigg(		 \frac{1}{m}			\bigg)^{	  \frac{ n + 1 }{2 }  		} 
		\bigg)	 
		\leq 
		C n  \frac{	 \Vu		}{\mu^{1/2}}
		\label{estimate:largemass}
	\end{equation}
where in the last inequality we use  that  for all $k \in \N$    in the large mass regime: 
\begin{equation}
	\label{mass}
	\frac{\Vu^{k }}{\mu^{  \frac{k}{2}	}}   
	\bigg(	 \frac{1}{m}		\bigg)^{	\frac{k -1}{2}	}   
 \ 	= \ 
	\bigg(	 \frac{	\Vu	}{ \mu^{1/2}m}  	\bigg)^{ \frac{k-1}{2}}
	\bigg(	 \frac{	\Vu	}{ \mu^{1/2}}  	\bigg)^{	 \frac{k+1}{2}	}
 \ 	\leq   \ 
	\frac{\Vu}{\mu^{1/2}}
\end{equation} 
 Thus, we	arrive at    (again, we absorb powers of $n$ in $C^n$): 
		\begin{align}
			\label{large mass case}
			\| 	\Psi^{(n)}_{	\eff}(t)    -  \Psi(t ) \|     
			& 
			\ 	\leq   \ 
  C^n \,  \sqrt{n!} \, 
			\bigg(
			\frac{\Vu}{\mu^{\frac{1}{2}}	}
			+
 \frac{ | 	t | 			}{\mu^{1/2}}
			\frac{ \Vu^{n+1 }}{   (	 \mu m	)^{  \frac{n-1}{2}}			 }
			\frac{	\Vu	}{\mu^{\frac{1}{2}} m }
			\bigg) \ . 
		\end{align}
		The proof of the theorem is finished
		once we combine \eqref{large mass case} and \eqref{small mass case}. 
	\end{proof}

\section{Resolvent estimates}
\label{section resolvent estim}
The main goal of this section
is to give a proof of Proposition \ref{prop rem}, stated    in Section \ref{subsection6.1}, 
which allows us to control the norm of states of the form \eqref{prod} in the case of massive fields.  
Thus,  we assume throughout    this section  
that the   dispersion relation  is   given by 
\begin{equation}
	\omega(k) =    \sqrt{ k^2 + m^2 } 
\end{equation}
with $m\ge \mu^{-1}$. The generator  is taken as in Definition \ref{def gn}, 
but the results in this Section apply to any non-negative bounded measurable function $g(p)  \geq  0$.

\vspace{1mm}

{ Throughout this section, 
the parameter  $0  \leq  \ve < 1/2$ is fixed, 
and $V$ is the form factor satisfying Condition \ref{condition 2}. 
} 
\vspace{1mm}

The organization of this section is as follows: 
In Subsection \ref{subsection resolvent}
we introduce 
relevant   resolvent functions.
In Subsection \ref{subsection example} 
we give  an example  of an estimate for a state \eqref{prod} 
consisting of only $\V^+$.
Finally, in Subsection \ref{subsection proof}
we turn to the proof of Proposition \ref{prop rem}.

\subsection{Resolvents}
\label{subsection resolvent}
Let us generalize the resolvents \eqref{resolvent 1}  and \eqref{resolvent 2}
that naturally emerged in the proof of Theorem \ref{thm 1}. 
To this end, let     $ n\geq 1$  and  $k_1 , \cdots, k_n \in \RE^3 $. 
We 
denote  $  \bk_n   \equiv  ( k_1, \ldots,  k_n)  \in \RE^{3n}$.  
For $  p \in \RE^3$, we introduce the following notations 
\begin{align}
	&  
\label{def Rp}
	\textstyle 
	R_p (\bk_n ) 
	: =
 (-1 ) 
	\bigg( \mu 
	\sum_{i=1}^n \omega (k_i) 
	+ 
	( p - \sum_{i=1 }^n k_i)^2 
	- p^2  + g(p)
	\bigg)^{-1 } 
\end{align}
Note that  $R_p(\bk_n)$
is symmetric with respect to the permutation of the variables $k_1, \cdots, k_n$. 
 
\vspace{2mm}

The following      estimates will be crucial in our analysis.  
We  remind the reader of
the notation 
$\omega_\mu (k) = \omega(k) + \mu^{-1  }$
and
$ 
\Vu 
\equiv \|	\omega_\mu^{-1 }	V\|_{L^2}$.

\begin{lemma}\label{lemma resolvent}
Let $ n \geq 1$. Then, 
the following   is  true for  all 
	$ |p| \leq \ve \mu$
	and  all 
	 $ 1  \leq i \leq n $:  There exists a constant $C>0$
such that 
for all 
  $ \bk_n   \in \RE^{3n}$  
 	\begin{equation}
 			\label{resolvent estimate 1}
 		| 	R_p(\bk_n)  | 
 		\leq 
 		\frac{ C }{ \mu 	  \omega_\mu(k_i) 		 }  , \qquad | 	R_p(\bk_n  ) |   
 		\leq 
 		\frac{C }{ n m \mu } \  . 
 	\end{equation}
 
\end{lemma}

\begin{proof}
	The denominator on the right hand side of \eqref{def Rp} 
	can be  bounded below  using $ p \cdot k \leq \ve \mu |k|$ 
	and $g(p) \geq $ . We find that 
	\begin{align}
\textstyle
\nonumber  
		 \mu \sum_{j =1}^n \omega (k_j ) 
		+ 
		( p - \sum_{j=1 }^n k_j )^2 
		- p^2  + g(p)    
	& 	 =  
	\textstyle
    	 \sum_{j =1}^n  \big( \mu \, \omega (k_j )  -2 p \cdot k_j \big) 
		 + 
		 (  \sum_{j=1 }^n   k_j )^2 
   +  g(p)   \\  
  &  \textstyle \geq   C 	    \sum_{j =1}^n  \mu \omega (k_j )  \ . 
   \label{eq81}
	\end{align}
Next, 
note  that 
 we can use the bound $m \geq \mu^{-1}$
and  obtain that for constants $C_1$ and $C_2$
\begin{equation}
	C_1 \omega_\mu(k) \leq \omega (k) \leq C_2 \omega_\mu (k) \ . 
\end{equation}
Note also   that  \eqref{eq81} implies   $R_p(\bk) < 0$. 
This finishes the proof of the first bound in \eqref{resolvent estimate 1}, 
after  we use the trivial bound 
$ \sum_{j =1}^n   \omega_\mu  (k_j )\geq  \omega_\mu 	 (k_i) $.
The proof of the second bound in \eqref{resolvent estimate 1} is analogous, but uses the alternative bound  $\sum_{j =1}^n   \omega_\mu  (k_j )\geq  n m $ .
\end{proof}

\begin{remark}[Heuristics]
We are interested in the regime for which   
	 $  \omega_\mu^{-1 } V \in L^2$ 
	grows much slower than $\mu>0$ (for instance, logarithmically as in the Nelson model). 
	Consequently, 
	the generic bound that we use for the  
	resolvent function  is of the form  (here $n=1$)
\begin{align*}
  \mu^{1/2}  \|	R_p(k)  V(k)  	\|_{L^2 } \leq  \frac{C}{\mu^{1/2}} \Vu 
  \end{align*}
uniformly in $| p |\leq \ve \mu.$
	Thus, up to the $\Vu $ norm, 
	this  estimate extracts a \textit{full} factor $\mu^{-1/2}$.
	Hence, we refer to this estimate as a \textit{good resolvent bound}.

	On the other hand, in the presence of contractions of creation and annihilation operators, we will use the generic bound  (here $n=2$) 
\begin{align*}\label{eq:bad:res:estimate}
 \mu^{1/2}
	\bigg\| \int_{\R^3}  R_p	(k_1) 			    V(k_2) 					\psi (k_1 , k_2)  \d k_2  	\bigg\|_{L^2(\R_{k_1}^3) } 
	\leq 
	\frac{		C	}{		 \mu^{1/2} m } 
	\|	 	 V			\|_{L^2 }
	\|	 \psi			\|_{L^2}	
\end{align*}	
for $\psi  \in L^2(\RE^{6})$, which is only able to extract the possibly \textit{larger}  factor 
	$ \mu^{-1/2}m^{-1}$,
	uniformly in $|p| \leq \ve\mu.$
	Hence, we refer to these as \textit{bad resolvent bounds}. 
\end{remark}

\begin{remark}[Operator calculus]\label{remark operator}
	 The resolvent functions $R_p(\bk_n)$ 
	 define bounded operators  on  the spectral subspace $\1 (|p| \leq \ve \mu) L_X^2$ via 
		spectral calculus of $p = - i  \nabla_X$, 
		which we extend by zero to $L_X^2$. 
		Unless confusion arises
		we also denote these operators by $R_p( \bk_n)$.
		In particular, the estimates \eqref{resolvent estimate 1}
		are translated into operator norm estimates. 
\end{remark}

\begin{remark}\label{remark massless}
	The first bound in \eqref{resolvent estimate 1} 
	also apply in the massless case $m=0$
	provided the generator satisfies the lower bound $   g(p) \geq  C $.
	This only has the effect of changing the value of the constant in Lemma \ref{lemma resolvent}, and 
	we will use this observation in Appendix \ref{section mild}. 
\end{remark}

\subsection{An example}
\label{subsection example}
Before  we turn to the   estimates
of a general state of the form \eqref{prod}, 
let us warm up with an example that requires only good resolvent bounds. 
This state arises as the case where only $\V^+$  operators 
are considered, as no contractions are present. 

\vspace{1mm}

For  $ n \geq 1$ one can calculate the following  representation for the $n$-body state 
using commutation relations  
\begin{align}
	( 	\R \V^+  \cdots \R \V^+ )  (\vp \otimes \Omega)
	&  = 
	\mu^\frac{n}{2}
	\int_{\RE^{3k}} 
	e^{ - i  \sum_i k_i \cdot  X}
	\prod_{i=1}^n  V (k_i)    R_p (\bk_i) 
	\vp \otimes 
	b_{k_n}^* \cdots b_{k_1}^* \Omega  \d k_1 \cdots \d k_n \  . 
\end{align}
In particular, thanks to Lemma \ref{lemma resolvent}
and the  number estimate \eqref{number estimate}
we find that  
\begin{align}
\label{n body}
	\|			( 	\R \V^+  \cdots \R \V^+ )  (\vp \otimes \Omega)	\|_\mathscr{H}^2
	& 
	 \ 	\leq \ 
	n !   \mu^n  \, 
	\int_{\RE^{3(n+1)}}
	\prod_{i=1}^n  |  V (k_i) |^2 \| R_p (\bk_i)   \vp   \|_{L_X^2}^2 
	\d k_1 \cdots \d k_n  	\nonumber \\
&  \ 	\leq   \ 
	\frac{ C^n    n ! }{     \mu ^{n}  }
	\Vu^{2n}
	\   . 
\end{align}
Here, $C>0$ is  a constant independent of $n$.

\subsection{The general case}
\label{subsection proof}
In the previous example, 
there were only \textit{good resolvent bounds} involved because only creation operators $\V^+$ appeared. 
When annihilation operators  $\V^-$ are present, one needs to use also \textit{bad resolvent estimates} of the form \eqref{eq:bad:res:estimate}.
For such terms, 
we carry on the contractions
between creation- and annihilation-
operators
to estimate them. We now study this process.

\vspace{2mm }
In what follows,  
we denote $\overline{B}_{\ve \mu }	\equiv \{  p \in \RE^3 :  |p| \leq \ve \mu\}$. 
Our first step is to introduce an appropriate class of functions 
with a relevant norm. 
Namely, for $ n \in \mathbb N $
we consider     
\begin{equation}
	\M   \in 
	L^2 L^\infty 
	(  \RE^{dn}	 \times \overline{B}_{\ve \mu }	 	) 
	\equiv  
	L^2		\big(	\RE^{3n}	 ; 			L^\infty (  \overline{B}_{\ve \mu }	 	) 	\big)  
\end{equation}
which we equip with the norm 
\begin{equation}
	\|	 \M 	\|_{L^2 L^\infty } 
	: = 
	\bigg( \int_{\RE^3 }	 \Big( 
	\sup_{|p | \leq \ve \mu }	 | \M (  \bk_n,p ) 	|	
	\Big)^2 \d  k  		\bigg)^{1/2} \ .
\end{equation}
Analogously as in Remark \ref{remark operator}, 
these functions induce  bounded operators $\M (\bk_n,p)$
on $L_X^2$.

\vspace{2mm }

\textit{Example}. 
Up to irrelevant factors, these functions
should be regarded as the coefficients of the product states 
\eqref{prod}. 
For instance, in 
the previous example we have
\begin{equation}
	\label{n prod}
	( 	\R \V^+  \cdots \R \V^+ )  (\vp \otimes \Omega)
	=  
	\int_{\RE^{3k}} 
	e^{ - i  \sum_i k_i \cdot  X}
	\M_n (\bk_n,p )
	\vp \otimes 
	b_{k_n}^* \cdots b_{k_1}^* \Omega  \d k_1 \cdots \d k_n \  . 
\end{equation}
with $\M_n ( \bk_n ,p ) : =		\mu^\frac{n}{2} \prod_{i=1}^n  V (k_i)    R_p  (\bk_i)   $.

\vspace{2mm}
The specific example 
given above can be readily generalized in the following sense.

\begin{definition}
	\label{def M state}
	Let $ n \in \mathbb  N $.
	Then, to any 
	$ \M \in	L^2 L^\infty 
	(  \RE^{dn}	 \times  \overline{B}_{\ve \mu }	  ) $
	we associate the $n$-particle state  
	\begin{equation}
		\Phi_\M : = 
		\int_{\RE^{3k}} 
		e^{ - i  \sum_i k_i \cdot  X}
		\M (  \bk_n ,p)
		\vp \otimes 
		b_{k_n}^* \cdots b_{k_1}^* \Omega  \d k_1 \cdots \d k_n \  . 
	\end{equation}

\end{definition}

\begin{remark}
	For any     $\M \in	L^2 L^\infty 
	(  \RE^{dn}	 \times  \overline{B}_{\ve \mu }	 	)  $ 
	the number estimate \eqref{number estimate}  shows that 	 
	\begin{equation}
		\label{M bound}
		\| \Phi_{\mathcal{M}}	\|_\mathscr{H}
		\leq 
		\sqrt{n!}  \, 
		\|   \mathcal{M} 	\|_{L^2 L^\infty}    \ . 
	\end{equation}
	The bound \eqref{M bound}	
	motivates the introduction 
	of the space  $	L^2 L^\infty 
	(  \RE^{dn}	 \times\overline{B}_{\ve \mu }	  )  $.
	Let us also mention here  that we have 
	included the pre-factor
	$	 (-1)^n  	 e^{- i  \sum k_i \cdot X}$ only  for convenience. 
\end{remark}

Our next goal is to
study the action of $\R\V^+$
and $\R\V^-$
on states  $\Phi_\M$.
This calculation is recorded  in the following lemma, 
which is a straightforward consequence of the
definitions of $\V^+$, $\V^-$, $\R$ 
and the commutation relations.

\begin{lemma}\label{lemma RV}
	Let  $ \M \in	L^2 L^\infty 
	(  \RE^{dn}	 \times \overline{B}_{\ve \mu }	  )  $  and
	$\Phi_\M $   be as in Definition \ref{def M state}. 
	Then, 
	\begin{enumerate}[leftmargin=*]
		\item 
		For all $ n \geq 1 $	it holds that 
		$\R \V^+ \Phi_\M  = \Phi_{\M^+}$  where 
		\begin{equation} 
			\label{M+}
			\M_+ ( \bk_{n+1} ,p ) 
			: =   \mu^{1/2}
			V(k_{n+1})
			R_p ( \bk_{n+1}  )   \M(  \bk_n  ,p ) \ . 
		\end{equation}

		\item 
		For all $ n \geq 2$  it holds that 
		$\R \V^- \Phi_\M  = \Phi_{\M^-}$  where 
		\begin{equation} 
			\label{M-}
			\M_- (  \bk_{n-1} ,p) 
			: =   
			\mu^{1/2 }
			\sum_{ i=1}^n 
			\int_{\RE^3}    V(k_{n})   		 R_p  (   \bk_{ n-1 }) 
			\M(   \pi_i   \bk_n ,p  )   \d k_n 		\  , 
		\end{equation}
		where   we denote $\bk_n \equiv  (\bk_{n-1},k_n ) \in \RE^{dn}$ 
		and $\pi_i \in S_n $ is the transposition  that changes $k_i$ and $k_n$, i.e. 
		$\pi_i  \bk_n = (k_1 \cdots k_n \cdots k_i)$.  
	\end{enumerate}
\end{lemma}

\begin{remark}
	Using the  notation  of   Lemma \ref{lemma RV},  
	we find  that  thanks to the resolvent estimates of Lemma \ref{lemma resolvent} 
	the following upper bounds  hold:  for $ n \geq 1 $
	\begin{equation}
		\label{M+ 1}
		\|	 \M_+		\|_{L^2 L^\infty}
		\leq 
		\frac{	 C 	}{\mu^{1/2}	}
		\Vu 
		\|	 \M	\|_{L^2 L^\infty}
	\end{equation}
	and for $ n \geq 2 $ (here we use $n/(n-1) \leq 2$)
	\begin{equation}
		\label{M- 1}
		\|	 \M_-		\|_{ L^2 L^\infty  }
		\leq 
		\frac{	  C  	}{\mu^{1/2} m	}
		\|	 V\|_{L^2}	  
		\|	 \M	\|_{L^2 L^\infty }     \  , 
	\end{equation}
	where for the $\M_-$ bound we  made use of  the  Cauchy-Schwarz inequality. 
\end{remark}

We now have all the ingredients
that we need to   
estimate  the  norm of a general state of the form \eqref{prod}, i.e.
we  turn to the proof of 
Proposition \ref{prop rem}. 
Let us remind the reader of the notation 
\begin{equation}
	\Sigma(n)  = 
	\textstyle 
	\{  \sigma \in \{+1  , -1 \}^n	 \  : \  	
	\sum_{i=1}^j  \sigma (i)  \geq 1 	 \ \forall  j = 1 , \ldots, n	\} \  , \qquad n \in  \mathbb N \ . 
\end{equation}
as well as   $\V^{+1} \equiv \V^{+}$, 
$\V^{-1 } \equiv \V^-$ and
when $\sigma$ is known from context 
$n_\pm  \equiv  | \{i :  \sigma(i) = \pm 1 \}|$.

\begin{proof}[Proof of Proposition \ref{prop rem}]
	Note that for all $\sigma \in \Sigma(n)$ it holds that 
	$\sigma(1) = \sigma(2) = +1$. Then, 	let us start the proof by considering     the  following two-body state (see e.g.  Definition \ref{def M state})
	\begin{equation}
		\R \V^+\R \V^+ (\vp \otimes \Omega) = 
		\Phi_{\M_2}  
		\qquad
		\t{with}
		\qquad 
		\M_2 ( \bk_2 ,p  )=   \mu V(k_1)  V (k_2)  R_p( k_1 , k_2) R_p(k_1). 
	\end{equation} 
	In particular,  
	$
	\|	 \M_2	\|_{L^2 L^\infty}
	\leq 
	C \mu^{-1 }	
	\Vu^2 
	$
	thanks to the resolvent estimate \eqref{resolvent estimate 1}. 
	
	\vspace{2mm}
	
Fix now $\sigma\in \Sigma(n)$. 	We define inductively 
	a sequence $(\M_{\ell})_{\ell=2}^n$ as follows.
	Namely, for $\ell =2 $ we set $\M_2$ as above.
	Assume that $\M_\ell$ for $\ell \geq 2$ has been defined.
	We now define  in  terms of the notation of Lemma \ref{lemma RV}
	\begin{equation}
		\M_{\ell+1} =  \begin{cases}
			(\M_\ell)_+  &  , \quad \sigma(\ell+1) = +1 \\
			(\M_\ell)_-   &  , \quad \sigma(\ell+1) =  - 1 \\
		\end{cases} 
	\end{equation}
	In particular, the construction of the function $\M_n$  is so that  
	\begin{equation}
		\prod_{i=1}^n	(\R  \V^{\sigma (i) })	 (\vp \otimes \Omega) = 
		\Phi_{\M_n}  \ . 
	\end{equation} 
	Further, with the kernels defined in this way, we realize thanks 
	to estimates 
	\eqref{M+ 1} and \eqref{M- 1}
	that 
	$	 \|	 \M_{\ell  }	\|_{L^2 L^\infty}
	\leq   
	f(\sigma, \ell , \mu )    \|	 \M_{\ell -1   }\|_{L^2 L^\infty}$
	where we denote 
	\begin{equation}
		\label{f}		f(\sigma, \ell , \mu ) : = 	 
\frac{C}{\mu^{1/2}}
		\begin{cases} 
			\Vu
			\qquad &  \t{if} \quad  \sigma(\ell ) = + 1 
			\\
m^{-1 }
			\|	  	 V	\|_{L^2}  
			\qquad &   \t{if} \quad  \sigma(\ell ) = -1 
			\\
		\end{cases}
	\end{equation}
	for all $  \ell   \geq  2$. 
	Here,   $C>0$ is chosen as the maximum constant between the 
	constants that appear in \eqref{M+ 1} and \eqref{M- 1}. 
	This estimate can be iterated from $\ell = n $ to $\ell = 3 $ and we obtain 
	\begin{equation}
		\|	 \M_{ n  }	\|_{L^2 L^\infty}
		\leq   
		\prod_{ r = 3}^{ n}		 f(\sigma,  r , \mu )    \|	 \M_{ 2   }\|_{L^2 L^\infty}   \ . 
	\end{equation}
	Finally, we notice that $\Phi_{\M_n}$
	is a state of $  M  :  = n_+  - n_-$ particles.
	Thus, 
	it suffices now to  use  the number estimate 
	$\|	 \Phi_{\M_n }	\| \leq 	 \sqrt{M  ! }	\|	 \M_n \|_{L^2 L^\infty}$,
	the initial bound  
	$
	\|	 \M_2	\|_{L^2 L^\infty}
	\leq 
	C \mu^{-1 }	
	\Vu^2
	$, 
	and the      formula \eqref{f}. 
\end{proof}

\appendix
\section{Effective dynamics for massless fields: Revisited}
\label{section mild}
In this section, we consider massless models
and  iterate the algorithm  in the proof of Theorem \ref{thm 1} one more time. 
Our goal is two-fold.  First, we show 
for the Nelson model 
that it is possible 
to improve the time scale given by Theorem \ref{thm 1} 
by a logarithmic factor.
Secondly, 
we study   massless models
whose form factors
are less singular
than the Nelson model.

\vspace{1mm}

The following family of models
serves as an archetype  for both situations
\begin{equation}
	V_a(k) = |k|^{-a } \1 (|k |\leq \Lambda) \label{Va}  \ , 
\end{equation}
where $a \in [0,1/2]$. 
It will be convenient 
to introduce  the following scale  of norms, quantifying infrared behaviour
\begin{equation}
	\n	V	\n_{s ,\mu}  : = 		 \|	  \omega_\mu^{-s }	V	\|_{L^2}  \  , \qquad s>0 \ . 
\end{equation}
Note that our old norm fits into this scale as follows $\|	V	\|_{\mu} = \| V	\|_{1, \mu}$. 
The reader should keep in mind  that the following estimates hold for
the models  \eqref{Va}
\begin{align}
\label{Va bounds}
 & 	\|	V_a	\|_{1,\mu} \leq C (\log \mu)^{1/2} 
 & &   \t{and}
  & &   	\|	 V_a\|_{2,\mu}  \leq C \mu 	
  &   &    \t{ for } a = 1/2	\\ 
	& \|	V_a	\|_{1,\mu} \leq C 		 
 &  & 	\t{and}
	 & &    		 \|	V_a	 \|_{2,\mu} \leq C \mu^{\frac{1}{2}+a} 
  & 	   & 	\t{ for } a \in [0,1/2) \ . 
\end{align}

In our next theorem, we study the dynamics of massless fields and obtain an approximation that is valid 
over slightly longer time scales, in comparison to the results of Theorem \ref{thm 1}. In order to keep the next statement to a reasonable length, 
we have chosen to   
consider   only the form factors  given by   \eqref{Va}, although our results apply 
to more general  interactions 
if one wishes to keep track of the norms $\|V\|_{s,\mu}$.

\begin{theorem}
	\label{thm 3}
Let $\omega(k)=|k|$  and	   $ V = V_a$ as in \eqref{Va} with $ a \in [0,1/2]$.
Let $\Psi_0$ satisfy Condition \ref{condition 1}.  
Then, 	there exists a constant $C>0$ 
	such that for all $t \in \RE $
	and   $\mu>0$
	large enough 
	\begin{align}
& 	\|	 \Psi (t)		-		\Psi_\eff (t) 	\|_\H 
\leq 
C 
\bigg(
\frac{\log\mu}{\mu}
\bigg)^{	1/2 }
(1 + |t|) 	 & & \t{ for } a = 1/2	\ , 	\\
& 
	\|	 \Psi (t)		-		\Psi_\eff (t) 	\|_\H 
\leq  
	 \frac{C}{\mu^{1-a}}
(1 + |t|) 
& & \t{ for } a  \in [0,1/2)	\ .  
	\end{align}
\end{theorem}

\begin{remark} 
Let us now comment on the consequences that Theorem \ref{thm 3} has for the observed time scales. 
\begin{enumerate}[leftmargin=*]
		\item 
	For $a = \frac{1}{2}$ (the massless Nelson model)
 we have an effective approximation
for time scales
 $$\, 		 | 	t | \,  \ll  		\Big(  \frac{\mu}{\log\mu }	\Big)^{1/2}$$
which  improves the result of  Theorem \ref{thm 1}
by a factor $(\log \mu )^\frac{1}{2}$. 
\item 
	For $ 0 \leq  a <  \frac{1}{2}$ 
 we have
an effective approximation for time scales 
$$\, 		 | 	t | \,  \ll \mu^{1-a  }\   $$
	which improves the result of Theorem \ref{thm 1} by a factor $\mu^{1/2-a}. $
	\item 
The reader may wonder 
if it possible to extend the massless approximation to even longer time scales.
Currently, our methods do not allow 
for such an extension unless
one introduces a suitable infrared cut-off on the interaction potential $V$. In particular, we cannot prove a result similarly to Theorem \ref{thm 2} for the massless case. 
\end{enumerate} 
\end{remark}

\begin{proof}
	Similarly as in the proof of
	Theorem \ref{thm 1}, 
	we consider  	
	\begin{equation} 
		\|	 \Psi(t)		-		\Psi_\eff (t)	 \| 
		\leq  
		\|    \B(t) \R \V \Psi_0		\|
		+ 
		\|	 \I(t)	\Q_\Omega \V \R\V \Psi_0 	\|   
		\leq 
		\frac{ C \n	V \n_\mu 		  }{\mu^{1/2}}  
		+ 
		\|	 \I(t)	\Q_\Omega \V \R\V \Psi_0 	\|    , 
		\label{eq 3}
	\end{equation}
	where we have employed  Proposition \ref{prop B}
	to estimate the first boundary term. 
	Our next goal is to give a more precise estimate of the error term
	\begin{equation}
		\label{E}
		\E (t)  \equiv 		\I(t)	 \Q_\Omega \V \R \V \Psi_0   = 	\I(t)	  \V^+ \R \V^+ \Psi_0   \  , 
	\end{equation}
	where in the second line we have 
	kept only the non-zero contributions. 
	Making use of the integration by parts formula  \eqref{int by part formula} we find 
	\begin{align}
		\nonumber
		\E (t)   
		& = 	\I(t)	  \V^+ \R \V^+ \Psi_0  		\\
		\nonumber
		& =   \Big(   \B(t) \R   + \I(t) \V \R  \Big)   \V^+ \R \V^+ \Psi_0  	 \\ 
		\nonumber
		& = 
		\B(t) \R   \V^+ \R \V^+ \Psi_0 
		+ \I(t)  \V \R    \V^+ \R \V^+ \Psi_0  	  \\ 
		& = 
		\B(t) \R   \V^+ \R \V^+ \Psi_0 
		+ \I(t)  \V^+ \R    \V^+ \R \V^+ \Psi_0 
		+ \I(t)  \V^- \R    \V^+ \R \V^+ \Psi_0  \ . 
		\label{E eq 1}
	\end{align}
	Observe that the last term of the right hand side of \eqref{E eq 1} contains an operator $\V^-$. 
	This has the effect of introducing contractions between creation and annihilation operator, 
	which lead to worse infrared behaviour. 
	However, these are still  one-particle states  and 
	may be integrated by parts one more time. Namely, we find
	\begin{align}
		\nonumber
		\E (t)   
		& = 
		\B(t) \R   \V^+ \R \V^+ \Psi_0 
		+ \I(t)  \V^+ \R    \V^+ \R \V^+ \Psi_0 
		+ \B(t) \R   \V^- \R    \V^+ \R \V^+ \Psi_0   \\ 
		\nonumber
		& \quad + 
		\I(t) \V \R  \V^- \R    \V^+ \R \V^+ \Psi_0 \\ 
		\nonumber
		& = 
		\B(t) \R   \V^+ \R \V^+ \Psi_0 
		+ \I(t)  \V^+ \R    \V^+ \R \V^+ \Psi_0 
		+ \B(t) \R   \V^- \R    \V^+ \R \V^+ \Psi_0   \\ 
		&  \quad + 
		\I(t) \V^+ \R  \V^- \R    \V^+ \R \V^+ \Psi_0 	+	
		\I(t) \V^- \R  \V^- \R    \V^+ \R \V^+ \Psi_0  
		\label{E eq 2}
	\end{align}
	where in the last line we decomposed $\V = \V^+ + \V^-$ . 
	Next, we
	make use of the observation form Remark \ref{remark obs}, 
	as well as the bounds $ \|	\B(t )\| \leq 2 $
	and $\|	\I(t)\| \leq |t|$
	to find  that  
	\begin{align}
\nonumber 
		\|	\E(t) 	\| 
	& 	\leq 
		C 
		\|			 \R \V^+ \R \V^+ \Psi_0	\|
		+  C (1 + \, 	 \mu^{1/2}	 | 	t |  ) 	 		\|		 \R \V^+ 	 \R \V^+ \R \V^+ \Psi_0	\|
		 + C 		\|		 \R \V^- 	 \R \V^+ \R \V^+ \Psi_0	\|	\\
& 		 +  C  (1 + \, 			 \mu^{1/2}  | 	t |  )	\|	 \R \V^+ 	 \R \V^-  	 \R \V^+ \R \V^+ \Psi_0	\|
		 + C \, 		 | 	t |  \,  \|	  \V^- 	 \R \V^-  	 \R \V^+ \R \V^+ \Psi_0	\| 
	\end{align}
	for some constant $C>0$.
		All the terms in the expansion for $\E(t)$ are separately estimated 
	in Proposition \ref{prop E}.
	The proof of the theorem is finished once we 
	gather these estimates back in  \eqref{eq 3}
	and use  \eqref{Va bounds} to collect the leading order terms
	in the two different regimes for $ a \in[0,1/2]$. 
\end{proof}

\begin{proposition}\label{prop E}
	There is  a constant $C>0$ such that 	the following five estimates hold. 
	\begin{enumerate}[label=(\roman*),leftmargin=*]
		\item  
		$\|	\R   \V^+ \R \V^+ \Psi_0\|_\H 
		\leq 
		C\mu^{-1  } 
		\n V \n^2_{1,\mu }  
		$ 
		
		\item  
		$	\|	\R \V^+ \R    \V^+ \R \V^+ \Psi_0 				\|_\H 
		\leq 
		C \mu^{-3/2 }  	\n V \n_{1,\mu}^3  $

		\item 	
		$	\|	  \R   \V^- \R    \V^+ \R \V^+ \Psi_0	\|_\H 		
		\leq 
		C \mu^{-3/2 }
		\n V \n_{2 ,\mu}  			
		\n V \n_{1/2 ,\mu}^2 			  
		$

		\item 
		$\|	\R \V^+ \R  \V^- \R    \V^+ \R \V^+ \Psi_0 		\|_\H 	 
		\leq
		C \mu^{-2  }  
		\n V \n_{1,\mu }
	\n V \n_{2 ,\mu}  			
	\n V \n_{1/2 ,\mu}^2 	
		$

		\item 
		$	\|	 \V^- \R  \V^- \R    \V^+ \R \V^+ \Psi_0  		\|_\H 
		\leq
		C  \mu^{-1  }  
	\n V \n_{1 ,\mu}^2  
	\n V \n_{1/2 ,\mu}^2   \ . 
		$
	\end{enumerate}	
\end{proposition}

In the following proof we  will make use of   the   
estimates in Section \ref{section resolvent estim}
that were developed for massive fields,
 but apply  here  as well as long as no mass terms appear. 
 See Remark \ref{remark massless}. 
We also use the notation 
for resolvents $R_p(k)$, $R_p(k_1, k_2)$ etc.
introduced   in  Section \ref{subsection resolvent}.

\begin{proof}[Proof of Proposition \ref{prop E}]
		\textit{(i)} 
 Apply   \eqref{n body} with $n=2$.

	\noindent 	\textit{(ii)} 
	 Apply   \eqref{n body} with $n=3 $.  
	
\noindent 	\textit{(iii)}
	The state can be calculated explicitly using standard  commutation relations
	as well as the contraction of the field operators
 $		b_{k_3} b_{k_2}^*  b_{k_1}^* \Omega = 
\delta (k_3 - k_2)
b_{k_1}^* \Omega 
+ 
\delta (k_3 - k_1)
b_{k_2}^* \Omega  $.   We find that (here $\int = \int_{\RE^{6}} \d k_1 \d k_2$) 
	\begin{align} 
		\label{V-++}
		\R \V^-  	\R \V^+ \R \V^+ \Psi_0 				
		& \,  =  \, 
				\mu^{3/2}
		\int 		   V(k_1)	 | 	  V(k_2)  |^2
		e^{-i  k_1 X }
		R_p (k_1)^2 R_p ( k_1, k_2)	   	\vp \otimes  
		b_{k_1}^*	\Omega     
		\\ 
		& + 
				\mu^{3/2}
		\int 		 |   V(k_1) |^2	   	  V(k_2)  
		e^{-i  k_2 X }
		R_p (k_1)	R_p (k_2) R_p ( k_1, k_2)	  	\vp \otimes  
		b_{k_2}^*	\Omega   
	  \, 	=:		\,  \Phi_{\M_{1}}
		\nonumber
	\end{align}
Here  
	we have written 
 the one-particle state	$	\R \V^-  	\R \V^+ \R \V^+ \Psi_0 				$
 in terms of $\Phi_{\M_{1}}$ (see Definition \ref{def M state})
 with coefficients
 \begin{align}
\nonumber 
 	\M_{1}    ( k_1,p)   \ : =   \ 
&  	\mu^{3/2 }V(k_1)R_p (k_1)^2 
 	\int |V(k_2)|^2 R_p(k_1,k_2) \d k_2 \\ 
&  	+
 		\mu^{3/2 }
 	 	V(k_1)R_p (k_1) 
 	\int |V(k_2)|^2  R_p(k_2) R_p(k_1,k_2)   \d k_2  \ . \label{M1}  
 \end{align}
	Finally, we   borrow  estimates from Section \ref{section resolvent estim}.
	An application 
	of \eqref{M bound} and \eqref{resolvent estimate 1} imply that 
	\begin{equation}
		\|		\R \V^-  	\R \V^+ \R \V^+ \Psi_0 				\|
		\leq  \|	 	\M_{1}	\|_{L^2 L^\infty}
		\leq 	      
C \mu^{-3 /2 }
		\|	 \omega_\mu^{-2}	  V \|_{L^2}
		\|	   \omega_\mu^{-\frac{1}{2}}		  V \|_{L^2}^2    \ . 
		\label{M-++}
	\end{equation}

	\noindent 	\textit{(iv)}
In the notation of \textit{(iii)}  and Lemma \ref{lemma RV} we note that 
$ \R  \V^+ \R  \V^- \R    \V^+ \R \V^+ \Psi_0 	 = 	 \R \V^+ \Phi_{\M_{1}}	= \Phi_{ \M_1^+	} $. 
Consequently, 
\begin{equation}
			\|	  \R  \V^+ \R  \V^- \R    \V^+ \R \V^+ \Psi_0	\|	
 =   
 \|	 \Phi_{	 \M_1^+	} 	\|
 \leq 
 \|		 \M_1^+  \|_{L^2 L^\infty }
			\leq 
C \mu^{-1/2  }
	\n V \n_{1,\mu}
 \|	 {\M_{1}}		\|_{L^2 L^\infty} \ . 
\end{equation}
It suffices to combine the previous estimate with \eqref{M-++}
to finish the proof.

	\vspace{2mm}
	\noindent 	\textit{(v)}  
	We use the representation \eqref{V-++}
	and   act on it with  $\V^- = \int   V(k) e^{ikX}b_k \d k $  to find
	\begin{equation}
		 \V^- \R  \V^- \R    \V^+ \R \V^+ \Psi_0    = 
		 \mu^{\frac{1}{2 }}
		 \bigg(	 \int_{\RE^3} V(k) \M_1( k,p) \d k 	\bigg) \vp \otimes \Omega \ . 
	\end{equation}
Next, we note that for all $|p| \leq \ve \mu$ 
we may estimate thanks to the explicit expression \eqref{M1} 
and the resolvent estimates \eqref{resolvent estimate 1}: 
\begin{equation}
\bigg\|
	\int_{\RE^3} 
	\mu^{1/2}
 V(k) \M_1 ( k,p) \d k 
\bigg\| 
 \leq 
 \frac{C}{\mu }
	\n V \n_{1,\mu}^2
	\n V \n_{1/2,\mu}^2 \ . 
\end{equation}
The proof is then finished once we combine  the last two  displayed estimates. 
\end{proof}

\vspace{3mm}

\noindent \textbf{Acknowledgments}. 
E.C. is deeply grateful to Robert Seiringer for his  hospitality at ISTA, without which this project  would not have been possible. E.C. is   thankful    to Thomas Chen for valuable comments and for pointing out useful  references. 
E.C  gratefully acknowledges support from the Provost’s Graduate Excellence Fellowship at The University of Texas at Austin and from the NSF grant DMS- 2009549, and the NSF grant DMS-2009800 through T. Chen.

\vspace{3mm}

\noindent \textbf{Data availability.} This manuscript has no associated data.
\vspace{3mm}

\noindent 
\textbf{Conflict of interest.} 
The authors state  that there is no conflict of interest.

\end{document}